%% file: content.tex
\documentclass{lmcs}
\pdfoutput=1

\usepackage[utf8]{inputenc}
\usepackage{lastpage}
\lmcsdoi{17}{4}{21}
\lmcsheading{}{\pageref{LastPage}}{}{}%
{Nov.~30,~2020}{Dec.~17,~2021}{}


\usepackage{algorithm}
\usepackage[noend]{algpseudocode}
\usepackage{amssymb}
\usepackage{amsmath}
\usepackage{graphicx}

\usepackage{mathpartir}
\usepackage{pifont} 
\usepackage{subcaption}
\usepackage{tikz}

\usepackage{macros}
\usepackage{mathabbrev}

\usetikzlibrary{matrix}
\usetikzlibrary{shapes}
\colorlet{lightgray}{gray!50}

\usepackage{comment}
\usepackage{todonotes}
\usepackage[normalem]{ulem}

\let\oldReturn\Return
\renewcommand{\Return}{\State\oldReturn}
\renewcommand{\gets}{:=}

\begin{document}

\title{Adaptive Non-linear Pattern Matching Automata}

\author[R.~Erkens]{Rick Erkens}
\address{Eindhoven University of Technology, De Groene Loper 5, 5612 AE, Eindhoven, The Netherlands}
\email{r.j.a.erkens@tue.nl}

\author[M.~Laveaux]{Maurice Laveaux}
\email{m.laveaux@tue.nl}

\begin{abstract}
  Efficient pattern matching is fundamental for practical term rewrite engines.
  By preprocessing the given patterns into a finite deterministic automaton the
  matching patterns can be decided in a single traversal of the relevant parts
  of the input term.  Most automaton-based techniques are restricted to linear
  patterns, where each variable occurs at most once, and require an additional
  post-processing step to check so-called variable consistency. However, we can
  show that interleaving the variable consistency and pattern matching phases 
  can reduce the number of required steps to find all matches. Therefore, we take 
  the existing adaptive pattern matching automata as introduced by Sekar et al 
  and extend these with consistency checks. We prove that the resulting deterministic
  pattern matching automaton is correct, and show several examples where some reduction can be achieved.
\end{abstract}

\maketitle

\section{Introduction}\label{sec:introduction}
Term rewriting is a universal model of computation that is used in various applications, for example to evaluate
equalities or simplify expressions in model checking and theorem proving.  In its simplest form, a binary relation on
\emph{terms}, which is described by the \emph{term rewrite system}, defines the available reduction steps.
Term rewriting is then the process of repeatedly applying these reduction steps when applicable.  The fundamental step
in finding which reduction steps are applicable is \emph{pattern matching}.

There are two variants for the pattern matching problem. \emph{Root pattern matching} can be described as follows:
given a term $t$ and a set of patterns, determine the subset of patterns such that these are (syntactically) equal to
$t$ under a suitable substitution for their variables.  The other variant, called \emph{complete pattern matching},
determines the matching patterns for all subterms of $t$.
Root pattern matching is often sufficient for term rewriting, because applying reduction steps can make matches found for subterms obsolete.
A root pattern matching algorithm can also be used to naively solve the complete pattern matching problem by applying it to every subterm.

As the matching patterns need to be decided at each reduction step, various \emph{term indexing}
techniques~\cite{Sekar01:indexing} have been proposed to determine matching patterns efficiently.
An adaptive pattern matching automaton~\cite{SekarRR95:adaptive}, abbreviated as APMA
(plural: APMAs), is a tree-like data structure that is constructed from a set of patterns.
By using such an automaton one can decide the matching patterns by only examining each
function symbol of the input term at most once.
Moreover, it allows for \emph{adaptive} strategies, \ie, matching strategies
that are fixed before construction such as the left-to-right traversal of the automata in~\cite{Graf91:left_to_right}.
The size of an APMA is worst-case exponential in the number of patterns.
However, in practice its size is typically smaller and this construction
step is beneficial when many terms have to be matched against a fixed pattern set.

The APMA approach works for sets of linear patterns,
that is, in every pattern every variable occurs at most once.
As mentioned in other literature~\cite{Graf91:left_to_right, SekarRR95:adaptive} 
the non-linear matching problem can be solved by first preprocessing the patterns,
then solving the linear matching problem and lastly checking so-called \emph{variable consistency}.
We can show that performing matching and consistency checking separately does not
minimise the amount of steps required to find all matches.
Therefore, we extend the existing APMAs to perform consistency checking as part of the matching process.
Our extension preserves the adaptive traversal of~\cite{SekarRR95:adaptive}
and allows information about the matching step to influence the consistency checking, and the other way around.
The influence of consistency checking on the matching step is only beneficial
in a setting where checking (syntactic) term equality,
which is necessary for consistency checking, can be performed in constant time.
This is typically the case in systems where (sub)terms are maximally shared,
which besides constant time equality checks also has the advantage of a compact representation of terms.

We introduce the notion of \emph{consistency automata}, abbreviated as CA (plural: CAs),
to perform the variable consistency check efficiently for a set of patterns.
The practical use of these automata is based on similar observations as the pattern matching automata.
Namely, there may be overlapping consistency constraints for multiple patterns in a set.
We prove the correctness for these consistency automata and provide an analysis of its time
and space complexity.
We prove that the consistency automaton approach yields a correct consistency checking algorithm for non-linear
patterns.
Finally, we introduce \emph{adaptive non-linear pattern matching automata} (ANPMAs),
a combination of adaptive pattern matching automata and consistency automata.
ANPMAs use information from both match and consistency checks to allow the removal of redundant steps.
We show that ANPMAs yield a correct matching algorithm for non-linear patterns.
To this end we also give a correctness proof for the APMA approach from
\cite{SekarRR95:adaptive}, which was not given in the original work.

\subsection{Structure of the paper}
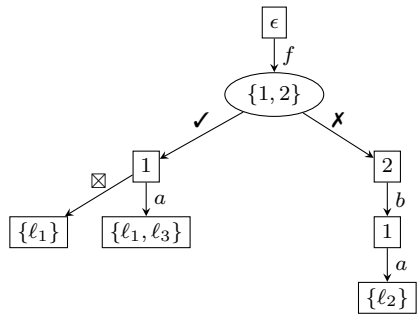
\begin{figure}[b]
\centering
\input{figures/example_introANPMA}
\caption{An example ANPMA.}
\label{fig:example_anpma}
\end{figure}
In Figure~\ref{fig:example_anpma} there is a simple example ANPMA
for the pattern set $\{\ell_1:f(x,x), \ell_2:f(a,b), \ell_3:f(a,a)\}$.
It has edges labelled with function symbols (coming from states labelled with positions)
and edges labelled with $\cmark$ or $\xmark$.
In Section~\ref{sec:anpma} we define ANPMAs formally.
Since we will treat a correctness proof, we first discuss both elements of ANPMAs.

\begin{itemize}
\item
In Section~\ref{sec:apma} we recall APMAs from \cite{SekarRR95:adaptive}.
We focus on formalities and give a correctness proof of this method which was not in the original work.
Correctness intuitively means that the automaton that is constructed for a set of linear patterns,
is suitable to efficiently decide the matching problem.
We also show that there are no redundancies in APMA evaluation.
That is, there is no state such that the same outgoing transition is taken,
no matter what the input term is.

\item
In Section~\ref{sec:ca} we define CAs,
again with a focus on formalities and correctness proofs.
Given a set of patterns, a set of consistency partitions can be computed.
The automaton that is constructed for this set of partitions,
is suitable to efficiently decide the variable consistency problem.
Although some redundancies can be removed,
it is still difficult to remove all redundancies from CAs.
\end{itemize}

These two sections provide a proper foundation for the formal details of ANPMAs in Section~\ref{sec:anpma},
since ANPMAs are automata that can have an interleaving between APMA states (with outgoing function symbol transitions)
and CA states (with outgoing $\cmark$- and $\xmark$-transitions).
We give the construction algorithm, a correctness proof based on the proofs of the previous section,
and show that ANPMAs are at least as efficient as first applying APMAs and then consistency checking,
or the other way around.

It is difficult to obtain an ANPMA without any redundancies.
We give some examples of pattern sets that are benefitted by consistency checks,
but where it is difficult to define a general construction procedure that can make use of these benefits.
Lastly we state the optimisation problem for ANPMAs such that it can be picked up for future work.

\subsection{Related Work}\label{sec:related}

We compare this work with other term indexing techniques.
Most techniques use tree-like data structures
with deterministic~\cite{Cardelli84, Graf91:left_to_right, SekarRR95:adaptive,
Weerdenburg07:match_trees} or non-deterministic~\cite{FessantM01:optimizing, Christian93:flatterms,
McCune1992:discrimination_trees, Voronkov95:code_trees, Graf1995:substitution_trees} evaluation.
In this setting a deterministic evaluation guarantees that all positions in the input term
are inspected at most once.
Non-deterministic approaches typically have smaller automata, but the same position might be inspected
multiple times for input terms as a result of backtracking.

Most of the mentioned techniques do not support matching non-linear patterns directly.
\emph{Discrimination trees}~\cite{McCune1992:discrimination_trees}, \emph{substitution trees}~\cite{Graf1995:substitution_trees} can be extended with on-the-fly consistency checks for matching non-linear patterns.
However, their evaluation strategy is restricted to left-to-right evaluation and variable consistency must be checked whenever a variable which has already been bound occurs
again at the position that is currently inspected in the term.
Our approach of introducing a state to check term equality is also present in \emph{match trees}~\cite{Weerdenburg07:match_trees}, \emph{code trees}~\cite{Voronkov95:code_trees} and the decision trees used in Dedukti~\cite{HondetB20}.
The main advantage of ANPMAs is that consistency checks are allowed to occur at any point in the automaton,
the evaluation strategy is not limited to a fixed strategy and there are fewer redundant checks,
which makes the matching time shorter.

\section{Preliminaries}\label{sec:preliminaries}
In this section the preliminaries of first-order terms and the root pattern matching problem are defined.
We denote the \emph{disjoint union} of two sets $A$ and $B$ by $A \uplus B$.
Given two sets $A$ and $B$ we use $A \to B$, $A \rightharpoonup B$ and
$A\hookrightarrow B$ to denote the sets of total, partial and total injective
functions from $A$ to $B$ respectively.  We assume that a partial
function yields a special symbol $\bot$ for elements in its domain for which it
is undefined.
Given a function $f:A\to B$ we use $f[a\mapsto b]$
to denote the mapping that satisfies $f(x)=f[a\mapsto b](x)$ if $x\neq a$ and $f(x)=b$ if $x=a$.

Let $\bF=\biguplus_{i\in\nats}\bF_i$ be a \emph{ranked} alphabet.
We say that $f \in \bF_i$ is a \emph{function symbol} with arity, written $\ar(f)$, equal to $i$.
Let $\Sigma = \vars \uplus \bF$ be a \emph{signature} where $\vars$ is a set of variables. 
The set of terms over $\Sigma$, denoted by $\terms_\Sigma$, is defined as the smallest set such that
$\vars \subseteq \terms_{\Sigma}$ and whenever $t_1, \dots, t_n \in \terms_{\Sigma}$ and
$f\in\bF_n$, then also $f(t_1, \dots, t_n) \in \terms_{\Sigma}$.
We typically use the symbols $x, y$ for variables, symbols $a, b$ for function symbols of arity zero (constants), $f, g, h$ for function symbols of other arities and $t, u$ for terms.
The \emph{head} of a term, written as $\rootsym$,
is defined as $\rootsym(x) = x$ for a variable $x$ and
$\rootsym(f(t_1, \dots, t_n)) = f$ for a term $f(t_1, \dots, t_n)$.
We use $\varsof(t)$ to denote the set of variables that occur in term $t$.
A term for which $\varsof(t) = \emptyset$ is called a \emph{ground term}.
A \emph{pattern} is a term of the form $f(t_1, \dots, t_n)$.
A pattern is \emph{linear} iff every variable occurs at most once in it.

We define the (syntactical) equality relation ${=}\subseteq \terms^2$ as the smallest relation such that
$x = x$ for all $x \in \vars$, and $f(t_1, \ldots, t_n) = f(t'_1, \ldots,
t'_n)$ if and only if $t_i = t'_i$ for all $1 \leq i \leq n$.
Furthermore, the
equality relation modulo variables $=_\dc\,\subseteq \terms^2$ is the smallest relation such
that $x =_\dc y$ for all $x, y \in \vars$, and $f(t_1, \ldots, t_n) =_\dc
f(t'_1, \ldots, t'_n)$ if and only if $t_i =_\dc t'_i$ for all $1 \leq i \leq
n$.  Both $=$ and $=_\dc$ satisfy reflexivity, symmetry and transitivity and
thus are equivalence relations, and we can observe that
${=}\subseteq{=_\dc}$.

A \emph{substitution} $\sigma$ is a total function from variables to terms.
The application of a substitution $\sigma$ to a term $t$, denoted by $t^\sigma$, is the term where variables of $t$ have been replaced by the term assigned by the substitution.
This can be inductively defined as $x^\sigma = \sigma(x)$ and $f(t_1, \ldots, t_n)^\sigma = f(t_1^\sigma, \ldots, t_n^\sigma)$.
We say that term $u$ \emph{matches} $t$, denoted by $t \below u$, iff there is a substitution $\sigma$ such that $t^\sigma = u$.
Terms $t$ and $u$ \emph{unify} iff there is a substitution $\sigma$
such that $t^\sigma = u^\sigma$.

We define the set of \emph{positions} $\positions$ as the set of finite sequences over natural numbers where
the \emph{root} position, denoted by $\emptypos$, is the identity element and concatenation, denoted by dot, is an
associative operator.  Given a term $t$ we define $t[\emptypos] = t$ and if $t[p] = f(t_1, \ldots, t_n)$ then $t[p.i]$
for $1 \leq i \leq n$ is equal to $t_i$.  Note that $t[p]$ may not be defined, \eg, $f(x,y)[3]$ and $f(x,y)[1.1]$
are not defined.
A position $p$ is \emph{higher} than $q$, denoted by $p \posleq q$, iff there is position $r \in \nats^*$ such that
$p.r=q$.
Position $p$ is \emph{strictly higher} than $q$, denoted by $p \posle q$, whenever $p \posleq q$ and $p \neq q$.
We say that a term $t[q]$ is a \emph{subterm} of $t[p]$ if $p \posle q$ and $t[q]$ is defined.
The replacement of the subterm at position $p$ by term $u$ in term $t$ is denoted by $t[p/u]$,
which is defined as $t[\emptypos / u] = u$ and
$f(t_1, \dots, t_n)[(i.p) / u] =f(t_1, \dots, t_i[p/u], \dots, t_n)$.
The \emph{fringe} of a term $t$, denoted by $\cF(t)$, is the set of all positions
at which a variable occurs, given by $\cF(t) = \{p\in\positions \mid t[p]\in\vars\}$.

We also define a restricted signature for terms with a one-to-one correspondence between variables and positions.
First, we define $\pvars$ as the set of \emph{position variables} $\{\dc_p \mid p\in\positions\}$.
Consider the signature $\pSigma = \bF \uplus \pvars$.
We say that a term $t \in \terms_{\pSigma}$ is \emph{position annotated} iff for all $p \in \cF(t)$ we have that
$t[p] = \dc_p$.
For example, the terms $\dc_\emptypos$ and $f(\dc_1,g(\dc_{2.1}))$ are position annotated,
whereas the terms $f(x)$ and $f(\dc_{1.1})$ are not.
Position annotated patterns are linear as each variable can occur at most once.

A \emph{matching function} decides for a given term and a set of patterns the exact subset
of these patterns that match the given term.

\begin{defi}\label{def:primitivematchingalgorithm}
  Let $\cL \subseteq \terms_{\Sigma} $ be a set of patterns.
  A function $match_{\cL} : \terms_{\Sigma} \rightarrow \pset{\terms_{\Sigma}}$ is
  a \emph{matching function for $\cL$} iff for all terms $t$ we have
  $match_{\cL}(t) = \{\ell \in \cL \mid \exists \sigma: \ell^\sigma = t\}$.
  If $\cL$ is a set of linear patterns then $match_{\cL}$ is a \emph{linear matching function}.
\end{defi}

\section{Adaptive Pattern Matching Automata}\label{sec:apma}
For a single linear pattern to match a given term it is necessary that every function symbol of the pattern
occurs at the same position in the given term.
This can be decided by a single traversal of the input pattern.

\begin{prop}\label{prop:altmatching}
Let $\ell$ be linear pattern and $t$ any term.
We have that $\ell \below t$ if and only if for all positions $p$: if $\hd(\ell[p])\in\bF$
then $\rootsym(\ell[p])=\rootsym(t[p])$.
\end{prop}

For linear patterns a naive matching algorithm follows directly from this proposition:
to find all matches for term $t$ one can check the requirement stated on positions in the proposition for every pattern separately.
However, for a \emph{set} of linear patterns we can observe that
whenever a specific position of the given term is inspected, a decision can be made for all patterns at the same time.
Exploiting these kind of observations to yield an efficient decision procedure
is the purpose of so-called \emph{term indexing techniques}~\cite{Sekar01:indexing}.
Sekar et al.~\cite{SekarRR95:adaptive} describe the construction of a so-called
\emph{adaptive pattern matching automaton}, abbreviated as APMA.
Given a set of \emph{linear} patterns $\cL$ an APMA can be constructed that can be
used to decide for every term $t \in \terms_\Sigma$ which patterns of $\cL$ are matches for $t$.
The main advantage of an APMA over other indexing techniques and the naive approach is that
every position of any input term is inspected at most once.

As an introduction to the techniques that are developed in later sections,
we recall the evaluation and construction procedures of APMAs.
The presentation that we use is slightly more formal compared to the presentation by Sekar et al.
The extra formalities provide a more convenient foundation for our extensions.
Moreover we present a correctness proof that did not appear in the original work,
which also mainly serves as a stepping stone to the correctness proofs for our extensions.

APMAs are state machines in which every state is a \emph{match} state, which is labelled with a position,
or \emph{final} state, which is labelled with a set of patterns.
Match states indicate that the term under evaluation is being inspected at the labelled position.
Final states indicate that a set of matching patterns is found.
The transitions are labelled by function symbols or an additional \emph{fresh} symbol
$\dcneq\,\notin \bF$. 
Let $\bF_{\dcneq} = \bF \uplus \{\dcneq\}$.

\begin{defi}
  An APMA is a tuple $(\states,\delta, \statelabel, s_0)$
  where:
  \begin{itemize}
    \item $\states=\statesfsym\uplus\statesfin$ is a finite set of states consisting of
      a set of \emph{match states} $\statesfsym$ and a set of \emph{final states} $\statesfin$;
    
    \item $\delta : \statesfsym \times \bF_{\dcneq} \rightharpoonup \states$ is a partial transition function;
    
    \item $\statelabel=\statelabelfsym\uplus\statelabelfin$ is a state labelling function with
    $\statelabelfsym:\statesfsym\rightarrow\positions$ and
    $\statelabelfin:\statesfin\rightarrow\pset{\terms_{\Sigma}}$; and
  
    \item $s_0 \in \statesfsym$ is the initial state.
  \end{itemize}
  We only consider APMAs that have a tree structure that is rooted in $s_0$.
  That is, $\delta$ is an injective partial mapping and there is no pair $(s,f)$ with $\delta(s,f)=s_0$.
\end{defi}

We illustrate the evaluation of an APMA by means of an example.
Consider the patterns $f(a, b, x), f(c, b, x) $ and $f(c, b, c)$ with $a, b, c \in \bF_0$, $f \in \bF_3$ and $x \in \vars$.
Figure~\ref{fig:example_APMA} shows an APMA that can be used to decide which of these patterns match for any given term.
Every state with outgoing transitions is a match state, and the other states at the bottom are the final states.
The match states are labelled with the position that is inspected during the evaluation of that state.

\begin{figure}
  \centering
  \input{figures/example_APMA}
  \caption{An APMA constructed from the patterns $f(a, b, x), f(c, b, x) $ and $f(c, b, c)$ with $a, b, c \in \bF_0$, $f \in \bF_3$ and $x \in \vars$.}
  \label{fig:example_APMA}
\end{figure}
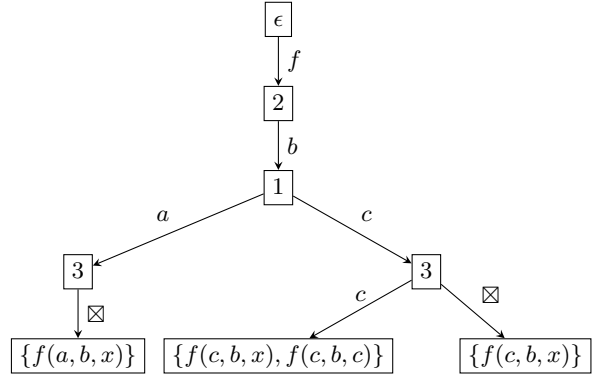

The evaluation of an APMA for a given term is defined in the function $\match$ defined in Algorithm~\ref{alg:apma:evaluation}.
Upon reaching a final state $s\in\statesfin$ the evaluation yields the set of terms $\statelabelfin(s)$, because these patterns match by construction.
In a match state $s\in\statesfsym$ the head symbol $\hd(t[\statelabelfsym(s)])$ is examined.
If there is an outgoing transition labelled with the examined head symbol then evaluation continues with the reached state; otherwise the $\dcneq$-transition is taken.
Whenever there is also no outgoing $\dcneq$-transition then there is no match by construction and the evaluation returns
the empty set as a result.

\begin{algorithm}
  \caption{Given a state $s$ of the APMA $M = (\states,\delta,\statelabel,s_0)$ and a term $t$, this algorithm computes the pattern matches for $t$ by evaluating $M$ on $t$.}\label{alg:apma:evaluation}
  \begin{equation*}
 \match(M, t, s) = 
  \begin{cases} 
    \statelabelfin(s)                  &\text{ if } s \in \statesfin \\
    \match(M, t, \delta(s,f))      &\text{ if } s \in \statesfsym \land \delta(s,f) \neq \bot  \\
    \match(M, t, \delta(s,\dcneq)) &\text{ if } s \in \statesfsym \land \delta(s,\dcneq) \neq \bot \land \delta(s,f) = \bot\\
    \emptyset                        &\text{ if } s \in \statesfsym \land \delta(s, \dcneq) = \delta(s,f) = \bot
    \end{cases}
 \end{equation*}
  {\hfill where $f = \rootsym(t[\statelabelfsym(s)])$}
\end{algorithm}

Consider the APMA $M$ of Figure~\ref{fig:example_APMA}
and let the initial state $s_0$ be the topmost state in the figure.
We have $\match(M,f(a,b,a),s_0) = \match(M, f(a,b,a), \delta(s_0, f)) = \ldots = \{f(a,b,x)\}$.
Similarly, we derive that $\match(M,f(b,b,b),s_0) = \emptyset$ and that
evaluating term $f(c,b,b)$ yields the pattern set $\{f(c,b,x)\}$.

The construction procedure for APMAs is defined in Algorithm~\ref{alg:constructapma}.
We use `\gets' to denote assignments to variables,
and we use $M[S\gets S']$ to denote that the element $S$ of the tuple $M$
gets updated to the value $S'$.
For example, $M\gets M[\statesfin\gets \statesfin\cup\{s\}]$ means that
$s$ is added to the set of final states in the APMA $M$.

Intuitively, the function $\constructapma$ constructs the APMA from the root state to final states
based on the set of patterns $\cL$ that could still result
in a match and a selection function $\select : \pset{\bP} \rightarrow \bP$.
For convenience we assume that the patterns in $\cL$ are position-annotated and as such these patterns are also linear.
In later sections we drop this assumption and treat arbitrary (non-linear) patterns.
The algorithm is initially called with the initial state $s_0$,
after which every recursive call corresponds to a state deeper in the tree.
The parameter $\select$ is a function that determines in each recursive call
which position from $\work$ becomes the label for the current state.
Based on the selected position, the current state and the pattern set,
outgoing transitions are created to fresh states where the construction continues recursively.

The parameter $\pref$ denotes the \emph{prefix} for a state $s$.
The function symbols in $\pref$ represent which function symbols have been matched so far
and the variables in $\pref$ represent which positions have not been inspected yet.
The special symbol $\dcneq$ is used to indicate that none of the patterns have the function symbol of the given term at that position.
For example, the prefix $f(\dc_1,\dc_2,\dc_3)$ represents that $f$ occurs at position $\epsilon$ of the input term and the variables at positions 1, 2 and 3 encode that these positions have not been inspected yet, or equivalently that these subterms are unknown.
The prefix can be reconstructed by following the transitions from $s_0$ to $s$.

Each recursive call starts by removing all the patterns from $\cL$ that do not unify with $\pref$.
Any match for the removed patterns cannot reach this state of the subautomaton that
is currently being constructed.
Therefore, the removed patterns do not have to be considered for the remainder of the construction.
If $\pref$ has the symbol $\dcneq$ at position $p$ then none of the
patterns in $\cL$ that have a non-variable subterm at position $p$ can unify with
the prefix any more, because $\dcneq$ does not occur in the patterns.
If there are no variables in $\pref$ then there is nothing to be inspected anymore.
This is the termination condition for the construction; the current state $s$ will be labelled with
the patterns that unify with $\pref$.
Otherwise, the work that still has to be done, \ie, the set of positions that still have to be inspected,
is the fringe of $\pref$, denoted by $\cF(\pref)$.

\begin{algorithm}
  \caption{Given a finite set of patterns $\cL$, this algorithm constructs an APMA for $\cL$.
    Initially, it is called with
    $M=(\emptyset,\emptyset,\emptyset,s_0)$,
    the initial state $s=s_0$ and the prefix $\pref=\dc_\epsilon$.
  }\label{alg:constructapma}
  \begin{algorithmic}[1]
  \Procedure{\constructapma}{$\cL, \select, M, s, \pref$}
    \State{$\cL' \gets \{\ell\in \cL \mid \text{$\ell$ unifies with $\pref$}\}$}
    \State $\work \gets \cF(\pref)$
      \label{line:constructapmalimit}
    \If{$\work = \emptyset$}
      \State $M \gets M[\statesfin := (\statesfin\cup\{s\}), \statelabelfin := \statelabelfin[s \mapsto \cL']]
          $\label{line:constructapmaterminate}
    \Else
      \State $\pos \gets \select(\work)$
      \State $M \gets M[\statesfsym:=(\statesfsym\cup\{s\}),\statelabelfsym := \statelabelfsym[s\mapsto\pos]]$
      \State $F \gets \{ f \in \bF \mid \exists \ell \in \cL': \rootsym(\ell[\pos]) = f \}$
      \For{$f \in F$}
        \State $M \gets M[\delta:=\delta[(s,f) \mapsto s']]$
            where $s'$ is a fresh unbranded state w.r.t. $M$
        \State $M \gets \constructapma(\cL, \select, M, s', \pref[\pos / f(\dc_{\pos.1},\dots,\dc_{\pos.\ar(f)}) ])$
      \EndFor
      \If{$\exists \ell\in \cL': \exists\pos'\posleq\pos:\rootsym(\ell[\pos'])\in\vars$}
        \State $M \gets M[\delta:=\delta[(s,\dcneq) \mapsto s']]$
          where $s'$ is a fresh unbranded state w.r.t. $M$
        \State $M \gets \constructapma(\cL, \select, M, s', \pref[\pos / \dcneq] )$
      \EndIf
    \EndIf
    \Return $M$
  \EndProcedure
  \end{algorithmic}
\end{algorithm}

\subsection{Proof of Correctness}
We prove that this construction yields an APMA that is suitable to solve the
matching problem for non-empty finite sets of linear patterns.
Meaning that the evaluation of the constructed APMA for any term $t$ is a linear matching function as defined in Definition~\ref{def:primitivematchingalgorithm}.

We make use of the following auxilliary definitions.
A \emph{path} to $s_n$ is a sequence of state and function symbol pairs
$(s_0,f_0),\dots,(s_{n-1},f_{n-1})\in\statesfsym\times\bF_{\dcneq}$
such that
$\delta(s_i,f_i)=s_{i+1}$ for all $i<n$.
Because $\delta$ is required to be an injective partial mapping there is a unique path to $s$ for every state
$s$, which we denote by $\pathof(s)$.
A match state $s$ is \emph{top-down} iff $\statelabel(s)=\emptypos$ or there is a pair $(s_i,f_i)$ in $\pathof(s)$
with $\statelabel(s_i).j=\statelabel(s)$ for some $1\leq j\leq \ar(f_i)$.
State $s$ is \emph{canonical} iff there are no two states in $\pathof(s)$ that are labelled with the same position.
Finally we say that an APMA is \emph{well-formed} iff all match states are top-down and canonical.

Well-formed APMAs allow us to inductively reconstruct the prefix of a state $s$
as it was created in the construction algorithm.
We allow slight overloading of the notation and denote the prefix of state $s$
by $\pref(s)$.
It is constructed inductively for well-formed APMAs by $\pref(s_0)=\dc_{\emptypos}$ and if
$\delta(s_i,f)=s_{i+1}$ then
$\pref(s_{i+1})=\pref(s_i)[\statelabel(s_i)/f(\dc_{\statelabel(s_i).1},\dots,\dc_{\statelabel(s_i).\ar(f)})]$.
Similarly, we denote the patterns of state $s$ by
$\cL(s)=\{\ell\in\cL\mid\text{$\ell$ unifies with $\pref(s)$}\}$ for all states.
Lastly we use an arbitrary function $\select: \pset{\positions} \rightarrow \positions$
such that for all sets of positions $\work$ we have $\select(\work) \in \work$.

\begin{lem}\label{lem:properconstruction}
  For all finite, non-empty sets of patterns $\cL$ we have that the procedure $\constructapma(\cL, \select, (\emptyset,\emptyset,\emptyset,s_0), s_0, \dc_\epsilon)$ terminates and yields a well-formed APMA $M=(\states,\delta,\statelabel,s_0)$.
\end{lem}
\begin{proof}
  Since $\cL$ is finite, the set $F$ that is computed on line 9 is also finite.
  Therefore the for loop only treats finitely many function symbols.
  It remains to show that the recursion terminates.
  The prefixes of the recursive calls are ordered
  by the strict matching ordering $\strbelow$.
  Observe that whenever $\pref[p]$ is defined,
  there must be a pattern $\ell\in\cL$ such that $\ell[p]$ is defined as well.
  There are only finitely many positions defined by the patterns of $\cL$.
  Therefore $\strbelow$ is a well-founded ordering on the recursive calls,
  which guarantees termination.
  
  Upon termination the result $M$ is indeed an APMA.
  For every function symbol in $F$ exactly one transition is created
  and at most one $\dcneq$-transition is created, so $\delta$ is a partial mapping.
  Since the target states of these transitions are fresh we have that $\delta$ is injective.
  Moreover there is no transition to $s_0$ since the algorithm is initially called with $s_0$.
  Hence $M$ is an APMA.
  
  We check that $M$ is well-formed.
  By construction we have $\statelabel(s_0)=\emptypos$ since the construction procedure is called
  with the prefix $\dc_\emptypos$.
  Let $s$ be an arbitrary non-final state and consider the stage of the construction algorithm
  $\constructapma(\cL, \select, M, s, \pref)$.
  A position label $p.i$ is only chosen if
  it occurs in the fringe of $\pref$.
  Therefore there must have been a state labelled with $p$
  where the variable $\dc_{p.i}$ was put in the prefix,
  so $s$ must be top-down.
  Lastly $s$ is canonical because once a position $p$ is chosen,
  it cannot be chosen again since the variable $\dc_p$ is replaced by an element of $\bF_{\dcneq}$ in the prefix.
  Hence $M$ meets all requirements for well-formedness.
\end{proof}

\noindent
For the remainder of the correctness proof
assume an arbitrary finite, non-empty set of position annotated patterns $\cL$
and let $M=(\states,\delta,\statelabel,s_0)$ be the APMA for $\cL$
that results from $\constructapma(\cL, \select, (\emptyset,\emptyset,\emptyset,s_0), s_0, \dc_\epsilon)$.

The following lemma states some claims about final states.
They are mostly necessities for the other lemmas.
A final state is characterised by a ground prefix and a non-empty set of patterns.

\begin{lem}\label{lem:finalstates}
  For every every final state $s$:
  (a) the set $\statelabel(s)$ is non-empty,
  (b) $\pref(s)$ is a ground term, and
  (c) for all $\ell\in\statelabel(s)$ we have $\ell\below\pref(s)$.
  Moreover (d) for every pattern $\ell\in\cL$ there is at least one final state $s$
  with $\ell\in\statelabel(s)$.
\end{lem}
\begin{proof}
  First observe that $\cL(s)$ is non-empty for all states $s$.
  Let $s$ be a final state.
  \begin{enumerate}[a)]
  \item
  Since $\statelabel(s)=\cL(s)$ and $\cL(s)$ is non-empty
  the claim holds.
  
  \item
  The prefix $\pref(s)$ is ground for final states $s$ because the construction
  only creates final states if $\pref(s)$ has no variables.
  
  \item
  By construction we have $\statelabel(s)=\{\ell\in\cL\mid\text{$\ell$ unifies with $\pref(s)$}\}$.
  Since $\pref(s)$ is ground we have that for all $\ell\in\statelabel(s)$ that $\ell\below\pref(s)$.
  
  \item
  Let $\ell\in\cL$.
  The following invariant holds for constructed APMAs.
  For all match states $s'$, if $\ell \in \cL(s')$ then there exist an $f$
  and an $s''$ such that $\delta(s',f) = s''$ and $\ell \in \cL(s'')$.
  From the fact that $\cL(s_0)=\cL$ it then follows that every pattern
  will end up in some final state.
  \qedhere
  \end{enumerate}
\end{proof}

Lemma~\ref{lem:matchapmainvariant} is an invariant that relates the construction algorithm
to the evaluation function $\match$.
It means that whenever $t$ matches $\ell$,
then $\ell$ will invariantly be in the set of patterns that is associated with the states that are visited
by $\match$.
The proof is an induction on the length of the path to the state under consideration.
The details of this proof can be found in the appendix.

\begin{lem}\label{lem:matchapmainvariant}
  Let $t$ be an arbitrary term and let $\cL_t=\{\ell\in\cL\mid\ell\below t\}$.
  For all states $s$ such that
  $\match(M, t, s_0)=\match(M, t, s)$
  it holds that
  $\cL_t\subseteq \cL(s)$.
\end{lem}

Lemma~\ref{lem:uniquestate} claims two straightforward correctness properties.
Firstly, if no pattern matches $t$, then the evaluation function $\match$ will yield the empty set of patterns.
Secondly, if at least one pattern matches $t$, then the evaluation function will reach a final state.
The proof details are in the appendix.

\begin{lem}\label{lem:uniquestate}
  Let $t$ be an arbitrary term and let $\cL_t=\{\ell\in\cL\mid\ell\below t\}$.
  \begin{enumerate}[a)]
  \item If $\cL_t=\emptyset$ then $\match(M, t, s_0) = \emptyset$;
  
  \item If $\cL_t\neq\emptyset$ then $\match(M, t, s_0) = \match(M, t, s_f)$ for some final state $s_f$.
  \end{enumerate}
\end{lem}

From the invariant claimed in Lemma~\ref{lem:matchapmainvariant},
it follows that all pattern matches of $t$ are returned by the evaluation.
The following lemma additionally claims the converse:
all patterns returned by the evaluation are indeed pattern matches for $t$.
A detailed proof can be found in the appendix.

\begin{lem}\label{lem:validmatches}
  Let $t$ be an arbitrary term and let $\cL_t=\{\ell\in\cL\mid\ell\below t\}$.
  If $\match(M, t, s_0) = \match(M, t, s_f)$ for some final state $s_f$
  then $\statelabel(s_f)=\cL_t$.
\end{lem}

From these lemmas the following correctness theorem follows.

\begin{thm}\label{thm:correctnessapma}
 The function $\lambda t.\match(M, t, s_0)$ is a linear matching function for pattern set $\cL$.
\end{thm}
\begin{proof}
  Let $t$ be an arbitrary term and let $\cL_t=\{\ell\in\cL\mid\ell\below t\}$.
  If $\cL_t=\emptyset$ then by Lemma~\ref{lem:uniquestate}
  we get that $\match(M, t, s_0)=\emptyset=\cL_t$ as required.
  If $\cL_t$ is non-empty then by Lemma~\ref{lem:uniquestate} we have that
  $\match(M, t, s_0) = \match(M, t, s_f)$ for some final state $s_f$.
  Then by definition of $\match$ we get
  $\match(M, t, s_f)=\statelabel(s_f)$.
  From Lemma~\ref{lem:validmatches} it follows that
  $\statelabel(s_f)=\cL_t$,
  by which we can conclude $\match(M, t, s_0)=\cL_t$.
  Hence $\lambda t.\match(M, t, s_0)$ is a linear matching function for $\cL$.
\end{proof}

\subsection{Redundancy}
Algorithm~\ref{alg:constructapma} follows a very simple kind of construction.
At every state, figure out what still needs to be observed and then choose one of the positions
from $\work$.
Sekar et al. already observed that one kind of redundancy can be completely removed from the computed set
$\work$.
If there is a position in $\work$ where no pattern in $\cL$ has a function symbol at that position, then there is nothing worthwhile to observe.
In such cases Algorithm~\ref{alg:constructapma} creates no outgoing function symbol transitions,
but it does create a $\dcneq$-transition.
However, by definition of $\matchapma$ this transition is then taken regardless of the input term.
The evaluation of this state always takes an unnecessary step.
In this case we call the state $\dcneq$-redundant.
Formally, we identify two types of redundancies.
\begin{defi}[Redundancy for match states]
Let $s$ be a match state and let $M=(S,\delta,L,s_0)$ be an APMA.
\begin{itemize}
\item Given $f\in\bF_\dcneq$, we say that $s$ is \emph{$f$-redundant} iff
for all terms $t$, whenever $\match(M,t,s_0)=\match(M,t,s)$, then
$\match(M,t,s)=\match(M,t,\delta(s,f))$.
\item We say that $s$ is \emph{dead} iff
for all terms $t$, whenever $\match(M,t,s_0)=\match(M,t,s)$, then
$\match(M,t,s)=\emptyset$.
\end{itemize}
\end{defi}

In Figure~\ref{fig:example_APMA} the leftmost state labelled by position $3$ is $\dcneq$-redundant.
The procedure $\matchapma$ will always take the $\dcneq$-transition upon reaching that state.
Avoiding this kind of redundancies reduces the number of steps needed to declare a match,
which yields a more efficient matching algorithm.
Algorithm~\ref{alg:constructapma} only allows $\dcneq$-redundancies, as described above,
so this notion of $f$-redundancy is more general than needed for this section.
In Section~\ref{sec:anpma} we will see that $f$-redundant states, with $f\in\bF$,
can occur due to interleaving with consistency checks.
To avoid redundancies, Sekar et al. included the following in their construction procedure.

\begin{lem}\label{lem:apmanonredundancy}
Consider Algorithm~\ref{alg:constructapma} where Line~\ref{line:constructapmalimit}
is replaced by
\[\work\gets \cF(\pref) \setminus \bigcap_{\ell\in\cL'} \{p\in\bP\mid \hd(\ell[p])\not\in\bF\}\,.\]
Then the correctness argument of Theorem~\ref{thm:correctnessapma} still applies,
and every state of every APMA resulting from the construction is not dead and not $f$-redundant, for every $f$.
\end{lem}
\begin{proof}
For correctness we only point out the one reparation that needs to be done to the proof of Lemma~\ref{lem:finalstates}.
It no longer holds that $\pref(s)$ is ground for every final state $s$.
So it only remains to show that $\ell\below\pref(s)$ for all final states.
The prefix of a final state can still have variables,
but then for every position $p\in\cF(\pref(s))$ the pattern $\ell\in\cL'$ has a variable at that position or a higher
position.
This means that we still have $\ell\below\pref(s)$ for every $\ell\in L(s)$.

To show that there are no redundancies,
let $s$ be a match state of $M$ and
let $t$ be a term such that $\match(M,t,s_0)=\match(M,t,s)$.
Since $M$ is canonical, we know that $L(s) \neq L(s')$ for all states $s'$ in $\pathof(s)$.
No position is observed twice,
thus we only need to observe non-redundancy locally per state.
\begin{itemize}
\item
If $s$ has at least one outgoing $f$-transition with $f\neq\dcneq$ then by definition
$\match(M,t,s)=\match(M,t,\delta(s,f))$.
In this case we can easily construct a term $t'$ with $t\neq t'$ and
$\match(M,t',s_0)=\match(M,t',s)$ such that one of the following holds:
\begin{itemize}
\item $\match(M,t',s)=\match(M,t',\delta(s,g))$, for some $g\in\bF$ with $f\neq g$;
\item $\match(M,t',s)=\match(M,t',\delta(s,\dcneq))$, if there is a $\dcneq$-transition from $s$
\item $\match(M,t',s)=\emptyset$, if $s$ has no outgoing transition save for the $f$-transition.
\end{itemize}
In all three cases the evaluation of $t'$ continues differently, so $s$ is not $f$-redundant.

\item Otherwise, if $s$ has no outgoing $f$-transition, then it could have an outgoing $\dcneq$-transition
and no other transitions.
By the definition of $\constructapma$ this can only occur if no pattern in $\cL(s)$ has a function symbol on position $L(s)$.
However, the alternative $\constructapma$ defined in this lemma ensures that the chosen position for $L(s)$
is not a part of $\work(s)$ and therefore this case cannot occur.

\item Otherwise, if $s$ has no outgoing transitions at all, then it must be that $s$ is dead.
This case can only occur if $\cL(s)$ is empty, which is impossible.
\qedhere
\end{itemize}
\end{proof}

So in the simple case of linear terms, an APMA can be constructed that has no redundancies.
In later sections we see that avoiding redundancies for the non-linear matching algorithm is difficult.

\subsection{Strategies}
In this section we recall some of the known results on the time and space complexity of
APMAs~\cite{SekarRR95:adaptive}.
These results also carry over to the extensions that we propose later on.
The reason for providing a selection function in the construction procedure is to define a strategy for the order in which the positions of the input term are inspected.
The chosen strategy influences both the size and the matching time of the resulting automaton.
For the size of an automaton it is quite natural to count the number of states.
For example, consider the APMA for the patterns in Figure~\ref{fig:example_APMA} where position one was chosen before position two.
The resulting APMA has nine states instead of eight.

There are heuristics described for the selection function based on local strategies, which only look at the prefix and the pattern set during the construction presented in~\cite{SekarRR95:adaptive}.
For example taking the position with the highest number of alternative function symbols in the patterns to maximise the breadth of the automaton.
However, these heuristics can all be shown to be inefficient in certain cases~\cite{SekarRR95:adaptive}.
Due to the following result it is unlikely that optimal choices can be made locally.

It was already known that minimisation of the number of states is NP-complete~\cite{SekarRR95:adaptive}.
This result is based on the NP-completeness of minimising the number of states for Trie indexing structures~\cite{ComerS76:trie_index}.
Alternatively, one can consider the \emph{breadth} of the resulting automaton, which is given by the number of final states, as the measurement of size.
The reason for this is that the total number of states is at most bounded by the height of the automaton times the number of final states.
The upper bound on breadth for the optimal selection function for a set of patterns $\{\ell_0, \ldots, \ell_n\}$ is $\functionorder{\prod_{i = 0}^n \length{\ell_i}}$~\cite{SekarRR95:adaptive}.
The lower bound on breadth for any selection function is $\functionorder{\alpha^{n-1}}$, where $\alpha$ is the average number of function symbols and $n$ the number of patterns.

The matching time of an APMA can be defined in two different ways.
We can consider a notion of \emph{average matching time} based on a distribution of input terms.
However, this information is not typically available and assuming a uniform distribution seems unrealistic in practice.
Therefore, we consider the upper bound on the matching time to be determined by the final state with the highest depth for the optimal selection function, and the lower bound is the lowest depth for any selection function.
Here, the depth of a state is defined by the length of the path to the root of the automaton.
For any pattern sets the upper bound on the depth is $S$, where $S$ is the total number of non-variable positions, and the lower bound on the depth is $\Omega(S)$~\cite{SekarRR95:adaptive}.

An alternative approach is to define a notion of relative efficiency that compares two APMAs.
For a given APMA $M$ and a given term $t$ we define the matching time as the \emph{evaluation depth}, denoted by $\evaldepth(M, t)$, based on the number of recursive $\match$ calls required to reach a final state.
Then one can consider a notion of relative efficiency that compares two APMAs.

\begin{defi}
  Given two APMAs $M = (\states, \delta, \statelabel, s_0)$ and $M' = (\states', \delta', \statelabel', s_0')$ for a set of patterns $P$.
  We say that $M \preceq M'$ iff for all terms $t \in \terms_\Sigma$ it holds that $\evaldepth(M, t) \leq \evaldepth(M', t)$.
\end{defi}

An interesting observation for the selection function is that choosing so-called \emph{index} positions always yields a
more efficient APMA as in the definition above and in the number of states.
These index positions are positions where all the patterns that can still match (the set $\cL$ in the construction algorithm) have a function symbol.
As such, taking these index positions (if they are available) always yields the optimal choice.
Furthermore, avoiding redundant states can also be shown to always yield a more efficient APMA.
Aside from these observations,
two selection functions can easily yield APMAs that are incomparable with respect to $\preceq$.
A number of heuristics for selection functions can be found in \cite{maranget:decisiontrees}.

\section{Consistency Automata}\label{sec:ca}

A linear matching algorithm can be used to solve the non-linear matching problem by renaming the patterns and checking so-called \emph{variable consistency} after the matching phase~\cite{Graf91:left_to_right, Sekar01:indexing, SekarRR95:adaptive}.
As a preprocessing step a renaming procedure is applied.
It renames each pattern to a linear pattern by introducing new variables.
The variable consistency check ensures that the newly introduced variables which correspond to variables in the non-linear pattern can be assigned a single value in the matching substitution.
This can be seen as a many-to-one context where multiple introduced variables correspond to a single variable in the original pattern.
For the non-linear matching algorithm we can first use a linear matching algorithm to determine matches for the renamed patterns.
Followed by a consistency check to remove the linear patterns for which the matching substitution is not valid for the original patterns.

\subsection{Pattern Renaming}

A straightforward way to achieve the renaming would be
to introduce new variables for each position in the fringe of each pattern.  
However, for patterns $f(x, a)$ and $f(x', y')$ the variables $x$ and $x'$
could be identical such that the assignment for $x$ (or equally $x'$)
yields a substitution for both patterns. 
We can use position annotated variables for this purpose,
which are identical for the same position in different patterns,
to obtain these overlapping assignments.

For the consistency check it is necessary to keep track of equality constraints between variables that correspond to a single variable before a non-linear pattern is renamed.
For this purpose we use the notion of \emph{consistency classes}~\cite{Sekar01:indexing}.
A consistency class is a set of positions that are expected to be equal in order to yield a match.

\begin{defi}\label{def:consistency_class}
  Given a term $t$ and a consistency class $C \subseteq \positions$ we say that $t$ is \emph{consistent}  with respect to $C$ if and only if $t[p] = t[q]$ for all $p, q \in C$.
\end{defi}

A pattern can give rise to multiple consistency classes.  For instance, consider
the pattern $f(x,x,y,y,y,z)$. Based on the occurrences of variables $x, y$ and $z$
we derive the three consistency classes $\{1,2\}$, $\{3, 4, 5\}$ and $\{6\}$.
This means that for the input term $t = f(t_1, \dots, t_6)$ that $t[1] = t[2]$ and $t[3] =
t[4] = t[5]$ must hold, and finally $t[i] = t[i]$ holds trivially for all $1
\leq i \leq 6$, for this term to be consistent w.r.t. these classes.
A set of disjoint consistency classes is referred to as a \emph{consistency partition}.
The notion of term consistency w.r.t. a consistency class is extended as
follows. A term $t$ is consistent with respect to a consistency partition $P$ iff $t$
is consistent with respect to $C$ for every $C \in P$.

First, we illustrate the renaming procedure by means of an example.
Consider three patterns $f(x, x, z)$, $f(x, y, x)$ and $f(x, x, x)$.  
After renaming we obtain the following pairs of a linear pattern and the corresponding consistency partition: $(f(\dc_1, \dc_2, \dc_3) , P_1)$, $(f(\dc_1, \dc_2, \dc_3), P_2)$ and $(f(\dc_1, \dc_2, \dc_3), P_3)$; with the consistency partitions $P_1 = \{ \{1, 2\} ,\{3\} \}$, $P_2 = \{ \{1, 3\}, \{2\}\}$ and $P_3 = \{\{1,2,3\}\}$. 
The term $f(a, a, b)$ matches $f(\dc_1, \dc_2, \dc_3)$ as witnessed by the substitution $\idsubst[\dc_1 \mapsto a, \dc_2 \mapsto a, \dc_3 \mapsto b]$, but $f(a,a,b)$ is only consistent w.r.t. partition $P_1$. 
We can verify that the given term only matches pattern $f(x, x, z)$.

We define a rename function that yields a position annotated term and a consistency partition over $\cF(t)$ for any given term.

\begin{defi}
  The term rename function $\rename : \terms_\Sigma \rightarrow (\terms_{\Sigma_\positions} \times \pset{\pset{\positions}})$ is defined as
  \begin{align*}
    \rename(t) = (\rename_1(t, \epsilon), \{\{p \in \positions \mid t[p] = x\} \mid x \in \varsof(t)\})
  \end{align*}
  where $\rename_1(t, \epsilon) : (\terms_\Sigma \times \positions) \rightarrow \terms_{\Sigma_\positions}$ renames the
  variables of the given term to position annotated variables as follows.  
  \begin{align*}
    \rename_1(x, p) &= \dc_p & \text{if } x \in \vars \\
    \rename_1(f(t_1, \ldots, t_n), p) &= f(\rename_1(t_1, p.1), \ldots, \rename_1(t_n, p.n))
  \end{align*}
\end{defi}

Note that for linear patterns, the renaming results in a position annotated term with just
trivial consistency partitions that only consist of singleton consistency classes.
We show a number of characteristic properties of the $\rename$ function
which are essential for the correctness of the described two-phase non-linear matching algorithm.

\begin{lem}\label{lem:rename}
  For all terms $t \in \terms_{\Sigma}$ if $(t', P) = \rename(t)$ then:
  \begin{itemize}
    \item $t =_\dc t'$, and
        
    \item for all $p \in \cF(t)$: $t'[p] = \dc_p$, and
    
    \item for all $u \in \terms_\Sigma$ it holds that $u$ matches $t$ if and only if $u$ matches $t'$ and $u$ is consistent w.r.t. $P$.
  \end{itemize}
\end{lem}
\begin{proof}
  We can show by induction on $t$ that $t =_\dc \rename_1(t, \epsilon)$ to prove the first statement.
  For the second statement let $p \in \cF(t)$.
  First, we can show that $t'[p] = \rename_1(t[p], p)$ by induction on position $p$.
  From $t[p] \in \vars$ it follows that $t'[p]$ is equal to $\dc_p$.
  
  For the last property let $P$ be equal to $\{\{p \in \positions \mid t[p] = x\} \mid x \in \varsof(t)\}$
  and let $u$ be an arbitrary term.
  Assume that $u$ is consistent w.r.t. $P$ and $u$ matches $t'$.
  The latter means that there is a substitution $\sigma$ such that $t'^\sigma = u$.
  It follows that for all positions $p \in \cF(t')$ that $\sigma(t'[p]) = u[p]$.
  As $u$ is consistent w.r.t. $P$ it means that for all $x \in \vars$ and $p, q \in \positions$ that if $t[p] = t[q] = x$ then $u[p] = u[q]$.
  Therefore, we can construct the substitution $\rho$ such that for all $p \in \cF(t)$ we assign $u[p]$ to $t[p]$, where the latter is some variable in $\varsof(t)$.
  The observation of consistency above lets us conclude that there is only one such substitution $\rho$.
  From $t =_\dc t'$ it follows that $t^\rho = t'^\sigma$ and as such $t^\rho = u$, which means that $u$ matches $t$.
  
  Otherwise, if $u$ matches $t$ then there is a substitution $\sigma$ such that $t^\sigma = u$.
  Let $\rho$ be the substitution such that for all positions $p \in \cF(t)$ we assign $\sigma(t[p])$ (which is equal to $u[p]$) to $\dc_p$.
  As $t'$ is linear it follows that each $\dc_p$ is assigned once and thus $\rho(\dc_p) = \sigma(t[p])$ by definition.  
  Again, from $t =_\dc t'$ it follows that $t'^\rho = t^\sigma$ and as such $u$ matches $t'$.
  Finally, for all positions $p$ and $q$ such that $t[p] = t[q] = x$ for variable $x \in \vars$ it follows that
  $u[p] = u[q] = \sigma(x)$.
  We can thus conclude that $u$ is consistent w.r.t. $P$.
\end{proof}

For the variable consistency check a straightforward implementation follows directly from Definition~\ref{def:consistency_class}.
Let $P = \{C_1, \ldots, C_n\}$ be a consistency partition.
For each consistency class $C_i$, for $1 \leq i \leq n$, there are $\length{C_i} - 1$ comparisons to perform, after which the consistency of a term w.r.t. $C_i$ is determined.
This can be extended to partitions by performing such a check for every consistency class in the given partition.
We use the function $\isconsistent(t, P)$ to denote this naive algorithm.
For a set of partitions $\{P_1, \ldots, P_m\}$ the (naive) consistency check requires exactly
$\sum_{1 \leq j \leq m} \sum_{C \in P_j} \length{C}-1$ comparisons if $t$ is consistent w.r.t. $P$.
Furthermore, it requires at least $m$ comparisons for any term $t$.

For the renaming procedure we must consider that the patterns $f(x, x)$ and $f(x, y)$ are both renamed to the linear
pattern $f(\dc_1, \dc_2)$.  However, then it is no longer possible to identify the
corresponding original pattern.
This can be solved by considering an indexed family of patterns, which is defined as follows.
We assume the existence of an index set $\cI$ and use $\cL \times \cI$ to denote
the indexed family of patterns with elements denoted by $\indexof{\ell}{i}$ for $\ell \in \cL$ and $i \in \cI$.
We adapt the rename function to preserve the index assigned to each pattern.
Now, when given an indexed linear pattern that resulted from renaming we can identify the corresponding original pattern by its index.
We now combine these results to obtain a non-linear matching algorithm.
The correctness of the following lemma follows directly from the third property of Lemma~\ref{lem:rename}.

\begin{lem}\label{lem:consistency_check}  
Let $\cL \subseteq \terms_\Sigma \times \cI$ be an indexed family of patterns and let $\cL_r \subseteq \pterms \times \pset{\pset{\positions}} \times \cI$ be the indexed family of renamed patterns and corresponding consistency partitions resulting from renaming; \ie, $\cL_r = \{\rename(\indexof{\ell}{i}) \mid \indexof{\ell}{i} \in \cL\}$.
  Let $\textsc{match-linear} : \terms_{\Sigma} \times \pset{\terms_{\Sigma}} \times \cI \rightarrow \pset{\terms_{\Sigma}} \times \cI$
  be a linear matching function that preserves indices.
  For any term $t \in \terms_\Sigma$ we define $\textsc{match} : (\terms_\Sigma \times (\pterms \times \pset{\pset{\positions}} \times \cI)) \rightarrow \terms_\Sigma$ as:
  \begin{equation*}
    \textsc{match}(t, \cL_r) = \{\ell \mid \indexof{\ell'}{i} \in \cL' \land \indexof{(\ell', P)}{i} \in \cL_r \land \isconsistent(t, P) \}
  \end{equation*}
  where $\cL'$ is equal to $\textsc{match-linear}(t, \{\indexof{\ell'}{i} \mid \indexof{(\ell', P)}{i} \in \cL_r\})$.
  The function $\textsc{match}$ is a matching function.
\end{lem}

\subsection{Consistency Automata}
In this section, we are going to focus on solving the consistency checking efficiently by exploiting overlapping partitions to obtain an efficient linear pattern matching algorithm.
This is inspired by the APMA where overlapping patterns are exploited to obtain an efficient matching algorithm.
Consider the consistency partitions $P_1 = \{ \{1, 2\} ,\{3\} \}$, $P_2 = \{ \{1, 3\}, \{2\}\}$ and $P_3
= \{\{1,2,3\}\}$ again.
In this case we can use the result of comparisons from overlapping partitions to determine the subset of all consistent partitions efficiently.
We would expect that at most three comparisons $t[1] = t[2]$, $t[2] = t[3]$ and $t[1] = t[3]$ would have to be performed to determine the consistent partitions. 

To exploit this, we define \emph{consistency automata} which are constructed from a set of consistency partitions.  
A consistency automaton, abbreviated by CA, is a state machine where every state is either a consistency state, which is labelled with a pair of positions, or a final state, which is labelled with set of partitions.
Each consistency state is labelled with a pair of positions that should be compared.  
Similar labelling is also present in other matching algorithms~\cite{Voronkov95:code_trees}.
The transitions of a CA are labelled with either $\cmark$ or $\xmark$ to indicate that the compared positions are equal or unequal respectively.  
The evaluation of a CA determines the consistency of a term w.r.t. a given set of partitions.

\begin{defi}
  A \emph{consistency automaton} is a tuple $(\states, \delta, L, s_0)$ where:
  \begin{itemize}[beginpenalty=99]
    \item $\states=\statescons\uplus\statesfin$ is a set of states consisting of a set of consistency states $\statescons$
        and a set of final states $\statesfin$;
    \item $\delta : (\statescons \times \{\cmark, \xmark\}) \rightarrow \states$ is a transition function;
    \item $\statelabel = \statelabelcons \uplus \statelabelfin$ is a state labelling function with
    $\statelabelcons: \statescons \to \positions^2$ and $\statelabelfin : \statesfin \rightarrow \pset{\cI}$;
    \item $s_0 \in \states$ is the initial state.
  \end{itemize}
\end{defi}

We show an example to illustrate the intuition behind the evaluation function of a CA.
Consider the consistency partitions $P_1 = \{ \{1, 2\}, \{3\} \}$, $P_2 = \{ \{1, 3\}, \{2\}\}$ and $P_3 = \{\{1,2,3\}\}$
again.
Figure~\ref{figure:example_ca} shows a CA that can be used to decide the consistency of a given term $t$ w.r.t. any of these partitions.
In the state labelled with $\{1, 2\}$, the subterms $t[1]$ and $t[2]$ are compared.
Whenever these are equal the evaluation continues with the $\cmark$-branch and it continues with the $\xmark$-branch otherwise.
If a final state (labelled with partitions) is reached then $t$ is consistent w.r.t. these partitions by construction.  
Afterwards, we show that redundant comparisons can be removed such that this example requires at most two comparisons.

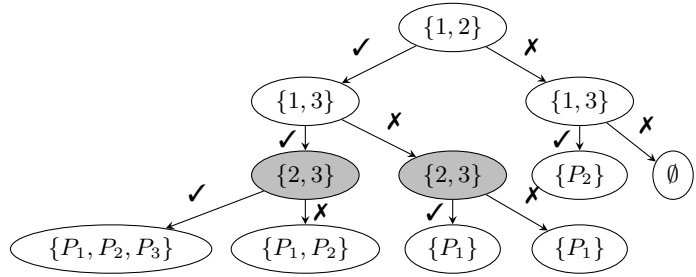
\begin{figure}
  \begin{center}
  \input{figures/example_CA_bad}
  \end{center}
  \caption{The CA for the partitions $P_1 = \{\{1, 2\},\{3\}\}$, $P_2 =
  \{\{1, 3\},\{2\}\}$ and $P_3=\{\{1,2,3\}\}$ where positions $1$ and $2$ are compared first, followed by $1$ and $3$ and finally $2$ and $3$.
  The grey states are redundant and can be removed as shown later on.}\label{figure:example_ca}
\end{figure} 

The evaluation function of a CA is defined in Algorithm~\ref{alg:ca:evaluation}.
If the current state is final then the label $\statelabel(s)$ indicates the set of indices such that $t$ is consistent w.r.t. the partitions $P_i$ for $i \in \statelabel(s)$.  
Otherwise, evaluation proceeds by considering the pair of positions given by $\statescons(s)$. 
The positions given by $\statescons(s)$ are unordered pairs of positions (or 2-sets), denoted by $\positions^2$, with elements $\{p, q\}$ such that $p \neq q$.  
These unordered pairs avoid unnecessary comparisons by the reflexivity and symmetry of term equality. 
If the comparison yields true,
the evaluation proceeds with the state of the outgoing $\cmark$-transition;
otherwise it proceeds with the state of the outgoing $\xmark$-transition.

\begin{algorithm}
  \caption{Given the CA $M = (\states, \delta, L, s_0)$ and a term $t \in \terms_\Sigma$ then $\evalca(M, t)$ returns the set of indices given by $\evalca(M, t, s_0)$ such that $t$ is consistent w.r.t. the corresponding partitions used for constructing $M$.}\label{alg:ca:evaluation}
  \begin{equation*}
    \evalca(M, t, s) = 
    \begin{cases}
      \statelabelfin(s)                &\text{ if } s \in \statesfin \\
      \evalca(M, t, \delta(s, \cmark)) &\text{ if } s \in \statescons \land t[p] = t[q] \text{ where $\{p, q\} = \statelabelcons(s)$}\\
      \evalca(M, t, \delta(s, \xmark)) &\text{ if } s \in \statescons \land t[p] \neq t[q] \text{ where $\{p, q\} = \statelabelcons(s)$}\\      
    \end{cases}
\end{equation*}
\end{algorithm}

The construction procedure of a CA is defined in Algorithm~\ref{alg:ca:construct}.  
Its parameters are the automaton $M$ that
has been constructed so far, the set of partitions $P$ and the current state
$s$.  Additionally, parameter $E$ contains the pairs of positions where the
subterms are known to be equal, and similarly $N$ is the set of pairs that are
known to be different. 
Lastly, a selection function $\selectca$ is used to define the strategy for choosing the next positions that are compared.

The partitions in $P$ for which a pair $\{p, q\}$ of positions is known to be
different are removed as these can not be consistent.
The remaining partitions form the set $P'$.
To denote the remaining work
concisely we introduce the notation $\subseteq\in$ for
the composition of $\subseteq$ and $\in$; formally $A \subseteq\in B$ iff
$\exists C \in B : A \subseteq C$.
Each pair of $E$ that has already been compared is removed from $\work$.
The condition on line~\ref{alg:construct_ca:base} checks whether there are no
choices left to be made.
If this is the case then all partitions in $P'$ are consistent by construction and the labelling function is set to yield the
partitions $P'$.

Otherwise, a pair $\{p, q\}$ of positions in $\work$ is chosen by the $\selectca$ function and two outgoing transitions are created.
A $\cmark$-transition is created that is taken during evaluation whenever the subterms at positions $p$ and $q$ are equal and this information is recorded in $E$.
Otherwise, the fact that these are not equal is recorded in $N$ and a corresponding $\xmark$-transition is created.

\begin{algorithm}
  \caption{Given a set of partitions $P = \{P_1, \ldots, P_n\}$ and a selection function $\selectca$ then $\constructca(P, \selectca)$ computes a CA using $\constructca(P, \selectca, (\emptyset, \emptyset, \emptyset, s_0), s_0, \emptyset, \emptyset)$ that can be used to evaluate the consistent partitions using $\evalca$.
    }\label{alg:ca:construct}
  \begin{algorithmic}[1]
    \Procedure{\constructca}{$P, \selectca, M, s, E, N$}
    \State $P' \gets \{P_i \in P \mid \neg\exists C \in P_i: \exists \{p, q\} \in N: p, q \in C\}$ 
    \State $\work \gets \{\{p, q\} \in \positions^2 \mid \{p, q\} \subseteq\in P_i \land P_i \in P'\} \setminus E$
    \If {$\work = \emptyset$}\label{alg:construct_ca:base}
      \State {$M \gets M[\statesfin \gets (\statesfin \cup \{s\}), \statelabelfin \gets \statelabelfin[s \mapsto P']]$}
    \Else   
      \State {$\{p, q\} \gets \selectca(\work)$}\label{alg:ca:construct:pick_variables}
      \State $M \gets M[\statescons \gets (\statescons \cup \{s\}), \statelabelcons \gets \statelabelcons[s \mapsto \{p, q\}]]$
      \State $M \gets \constructca(P, \selectca, M[\delta \gets \delta[(s, \cmark) \mapsto s']], s', E \cup \{\{p, q\}\}, N)$ where $s'$ is a fresh unbranded state w.r.t. $M$.
      \State $M \gets \constructca(P, \selectca, M[\delta \gets \delta[(s, \xmark) \mapsto s']], s', E, N \cup \{\{p, q\}\})$ where $s'$ is a fresh unbranded state w.r.t. $M$.
    \EndIf
    \Return $M$
    \EndProcedure
  \end{algorithmic}
\end{algorithm}

The consistency automata obtained from this construction are not optimal,
but later on we show that these redundancies can be removed.

\subsection{Proof of Correctness}
We show the correctness of the construction and evaluation of a CA as defined
in Theorem~\ref{theorem:ca:correct}. 
In the following statements let $P = \{P_1, \ldots, P_n\}$ be a set of partitions where each partition is a finite set of finite consistency classes and let $\select : \pset{\positions^2}
\rightarrow \positions^2$ be any selection
function such that $\select(\work) \in \work$ for all non-empty $\work \subseteq
\pset{\positions^2}$.
For the termination of the construction procedure we can show that the number of choices in $\work$ strictly decreases at each recursive call.

\begin{lem}\label{lem:ca:construct_terminates}
  The procedure $\constructca(P, \selectca)$ terminates.  
\end{lem}
\begin{proof}
  Consider the pair of positions $\{p, q\}$ that is taken from $\work$ at
  line~\ref{alg:ca:construct:pick_variables}.  It is easy to see that $\{p, q\}
  \notin E$, and $\{p, q\} \notin N$ follows directly from the fact that $P'$
  only consists of partitions of which the consistency classes do not contain
  positions together in a pair of $N$.  Therefore, it follows that in
  subsequent recursive calls $\{p, q\}$ cannot be in $\work$ again as either
  $E$ or $N$ is extended with $\{p, q\}$ and no elements are ever removed from
  $E$ or $N$.  Furthermore, the execution of all other statements terminates as
  $\nofpositions(P)$ is finite, which also means that $\length{E}$ and
  $\length{N}$ are finite as inserted pairs satisfy $\{p, q\} \subseteq\in P'$.
  Finally, the selection function terminates by assumption.
\end{proof}

For the construction procedure we can show that for parameter $s$ it holds that $s \notin \states$ as a precondition.
Therefore, we can use $\work(s) : \states \rightarrow \pset{\positions}$, $E(s) : \states \rightarrow \pset{\positions^2}$ and $N(s) : \states \rightarrow \pset{\positions^2}$ to denote the
values of $\work$, $E$ and $N$ respectively during the recursive call of
$\constructca(P, \select, M, s, E, N)$. 
For the termination of the evaluation procedure we can show that $\work(s)$ strictly decreases for the visited states.

For the proof of partial correctness we show a relation between the pairs in
$E(s)$ and $N(s)$ and the comparisons performed in the evaluation function.
First, we define for a term $t \in \terms_{\Sigma}$ and parameters $E, N \subseteq \pset{\positions^2}$ the notion of \emph{consistency} where $t$ is consistent w.r.t. $E$ and $N$, denoted by $(E, N) \models t$, iff:
\begin{itemize}
  \item $\forall \{p, q\} \in E : t[p] = t[q]$, and
  
  \item $\forall \{p, q\} \in N : t[p] \neq t[q]$
\end{itemize}

A consistency automaton $M = (\states, \delta, L, s_0)$ is \emph{well-formed} iff for all terms $t \in \terms_\Sigma$ and all recursive calls $\evalca(M, t, s_0) = \evalca(M, t, s_n)$ it holds that $(E(s_n), N(s_n)) \models t$.

\begin{lem}\label{lem:ca:consistency}
  Let $M = (\states, \delta, L, s_0)$ be the result of $\constructca(P, \selectca)$.
  Then $M$ is well-formed.
\end{lem}
\begin{proof}
  The recursive calls form an \emph{evaluation series} $(s_0, a_0),\ldots,(s_n, a_n)$ for $s_i \in \states$ and $a_i \in \{\cmark, \xmark\}$ for $0 \leq i < n$ such that $\evalca(M, s_i, t) = \evalca(M, s_{i+1}, t)$ and $\delta(s_i, a_i) = s_{i+1}$.
  Let $t \in \terms_{\Sigma}$ be any term.
  We prove the statement by induction on the length of the evaluation series.
  
  Base case. We have $E(s_0) = N(s_0) = \emptyset$ and as such the statement holds vacuously.
  
  Inductive step.
  Suppose that for $\evalca(M, t, s_0) = \evalca(M, t, s)$ the statement holds.
  Suppose that $\evalca(M, t, s) = \evalca(M, t, s')$ where $s' = \delta(s, a)$ for $a \in \{\cmark, \xmark\}$ and let $\statelabelcons(s) = \{p, q\}$.
  There are two cases to consider:
  \begin{itemize}
    \item $t[p] = t[q]$ in which case $E(s')$, where $s'$ is equal to $\delta(s, \cmark)$, is $E(s)$ extended with $\{p,q\}$ and $N(s') = N(s)$.
    
    \item Otherwise, $t[p] \neq t[q]$ in which case $N(s')$ is equal to $N(s)$ extended with $\{p, q\}$ and $E(s') = E(s)$.
  \end{itemize}
  In both cases $(E(s'), N(s')) \models t$ holds by definition.
\end{proof}

Finally, we can show the correctness of using consistency automata to evaluate the consistency of a given term w.r.t. partitions in $P$.

\begin{thm}\label{theorem:ca:correct}
  Let $M = (\states, \delta, L, s_0)$ be the result of $\constructca(P, \selectca)$.
  Consider an arbitrary term $t$ and suppose that $\evalca(M,s_0,t)=P'$.
  Then for all $P_j \in P$ it holds that $P_j \in P'$ iff the term $t$ is consistent w.r.t. $P_j$.
\end{thm}

\begin{proof}
  We have already shown termination of the construction procedure in Lemma~\ref{lem:ca:construct_terminates}.
  Let $P'$ be the set of partitions returned by $\evalca(M, t, s_0)$, let $P_i \in P$ be any partition and $\evalca(M, t, s_0) = \evalca(M, t, s_n)$ for some final state $s_n \in \statesfin$.
  By Lemma~\ref{lem:ca:consistency} it holds for all $\{p, q\} \in E(s_n)$ that $t[p] = t[q]$ and for all $\{p, q\} \in N(s_n)$ that $t[p] \neq t[q]$.
  \begin{itemize}[align=left]
    \item[$\implies$)] Assume that $P_j \in P'$.    
      For all $p, q$ such that $\{p, q\} \subseteq\in P_j$ it holds that $\{p, q\} \in E(s_n)$ as $\work(s_n)$ is equal to $\emptyset$ for $s_n$ to become a final state in the construction.
      Therefore, for all $p, q \in C$ for consistency class $C \in P_j$ it holds that $t[p] = t[q]$ and as such $t$ is consistent w.r.t. $P_j$.
    
    \item[$\impliedby$)] Assume that term $t$ is consistent w.r.t. $P_j$.
      Proof by contradiction, assume that $P_j \notin P'$.
      As such, there is a position pair $\{p, q\} \subseteq\in P_j$ such that $\{p, q\} \in N(s_n)$.
      However, then it follows that $t[p] \neq t[q]$, from which we conclude that $t$ can not be consistent w.r.t.~$P_j$.
      \qedhere
  \end{itemize}
\end{proof}

\subsection{Efficiency}

Similarly to APMAs we can consider the space and time measurements for CAs.
Given a CA $M = (\states, \delta, L, s_0)$ and a term $t$ we define the \emph{evaluation depth}, denoted by $\evaldepth(M, t)$, as the number of recursive $\evalca$ calls performed to reach the final state.
The size, denoted by $\length{M}$, is given by the number of states $\length{\states}$.
Note that here we refer to the total number of states for its size, instead of only the number of final states or equivalently its breadth.
We also define the notion of relative efficiency for CAs.

\begin{defi}
  Given two consistency automata $M = (\states, \delta, \statelabel, s_0)$ and $M' = (\states', \delta', \statelabel', s_0')$ for a set of consistency partitions $P$.
  We say that $M \preceq M'$ iff for all terms $t \in \terms_\Sigma$ it holds that $\evaldepth(M, t) \leq \evaldepth(M', t)$.
\end{defi}

We present two ways to improve the time and space efficiency of consistency automata.
First of all, the selection function used for construction influences the size and the evaluation time
for the resulting CA as shown in Figure~\ref{ca:selection}.
Choosing the pair of positions $\{2, 3\}$ before $\{1, 2\}$ in this example yields a CA
that is both larger in size and less efficient.
Therefore, it makes sense to consider heuristics for the selection function in practice.
For instance, taking the selection function that always picks positions $\{p, q\}$ to minimise the size of $\work$ in both branches would immediately yield the right CA in Figure~\ref{ca:selection}.
In cases where there is no overlap between the partitions the choice would be arbitrary.

\begin{figure}
\begin{minipage}{.59\textwidth}
  \input{figures/example_CA_select_bad}
\end{minipage}
\begin{minipage}{.4\textwidth}
  \input{figures/example_CA_select_good}      
\end{minipage}
\caption{Two CAs for the partitions $P_1 = \{\{1,2\}, \{3, 4\}\}$ and $P_2 = \{\{1, 2, 3\}\}$. The CA on the left chooses $\{2, 3\}$ first. However, as shown on the right selecting $\{1,2\}$ first removes both partitions, and leads to a smaller CA.}\label{ca:selection}
\end{figure}
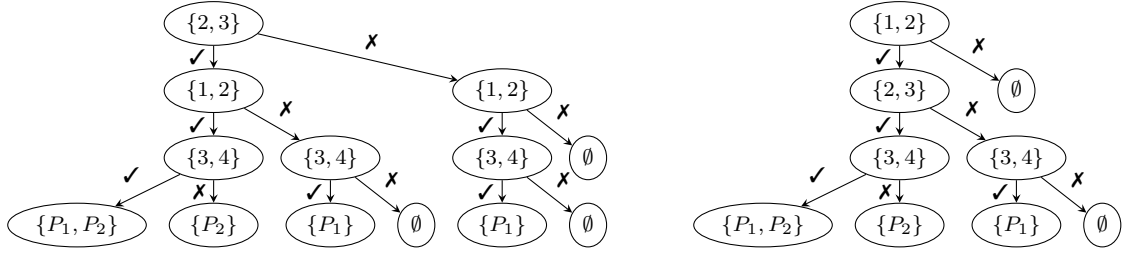 

Changing the selection function by itself does not necessarily result in the most relatively efficient CA.
If we consider Figure~\ref{figure:example_ca} again, we can observe that the resulting automaton is not optimal, despite being the smallest w.r.t. the selection function, because some of the final states are not reached during evaluation of any term.
For example, the final state labelled with $\{P_1, P_2\}$ is not reachable, because any term $t \in \terms_\Sigma$ that satisfies $t[1] = t[2]$ and $t[1] = t[3]$ can not have that $t[2] \neq t[3]$ by the transitivity of term equality.
Therefore, we can consider removing these states such that the evaluation depth of terms
where evaluation reaches these states is lower.
This is essentially the notion of redundancies that was previously defined for APMAs.

Given a CA $M = (\states, \delta, \statelabel, s_0)$ and a non-final state $s \in \statescons$ we give the following conditions for its redundancy.
\begin{defi}[Redundancy for CAs]
  Let $M = (\states, \delta, \statelabel, s_0)$ be a CA.
  \begin{itemize}
  \item A consistency state $s$ is $\cmark$-redundant iff for all terms $t$, whenever $\evalca(M, t, s_0) = \evalca(M, t, s)$, then $\evalca(M, t, s) = \evalca(M, t, \delta(s, \cmark))$.
    
   \item A consistency state  $s$ is $\xmark$-redundant iff for all terms $t$, whenever $\evalca(M, t, s_0) = \evalca(M, t, s)$, then $\evalca(M, t, s) = \evalca(M, t, \delta(s, \xmark))$.
  \end{itemize}  
\end{defi}

The notion of \emph{dead} states does not apply for CAs, because we explicitly keep final states labelled with the empty set for easier definitions.

Redundant states can be removed from the automata without affecting the correctness of its evaluation in the following way.
A state $s$ that is $\cmark$-redundant can be \emph{removed} by updating $\delta$ such that the incoming transition $\delta(r, a) = s$, for some $r \in \states$ and $a \in \{\cmark, \xmark\}$, is updated to $\delta(s, \cmark)$.
A similar transformation of $\delta$ can be applied for states that are $\xmark$-redundant using $\delta(s, \xmark)$.
We can observe that such a removal reduces the evaluation depth for terms where evaluation reached this state by one and that the size of the CA is reduced by the number of states in the $\cmark$-branch (or $\xmark$-branch) respectively if states unreachable by the transition relation are removed.
Next, we prove that removal does not influence the correctness of evaluation.

\begin{lem}\label{lem:wellformed-redundant}
  Let $M = (\states, \delta, L, s_0)$ be any CA that is well-formed.
  Then the resulting CA $M'$ where a $\cmark$-redundant or $\xmark$-redundant state $v \in \states$ is removed remains well-formed.
\end{lem}
\begin{proof}
  The recursive calls form an \emph{evaluation series} $(s_0, a_0),\ldots,(s_n, a_n)$ for $s_i \in \states$ and $a_i \in \{\cmark, \xmark\}$ for $0 \leq i < n$ such that $\evalca(M, s_i, t) = \evalca(M, s_{i+1}, t)$ and $\delta(s_i, a_i) = s_{i+1}$.
  By well-formedness of $M$ we know, for all terms $t \in \terms_\Sigma$ and all evaluation series $(s_0, a_0),\ldots,(s_k, a_k) \in (\states \times \{\cmark, \xmark\})$ of $\evalca(M, s_0, t)$, that for all states $s_i$, with $0 \leq i \leq k$, it holds that $(E(s_i), N(s_i)) \models t$.
  Now, we only have to consider sequences that contain the state $v$ as the other evaluation sequences remain the same.
  Consider any such sequence and let $u$ be the state in that sequence such that $\delta(u, a) = v$, for some $a \in \{\cmark, \xmark\}$, and let $t$ be an arbitrary term.
  Let $\{p, q\}$ be the value of $\statelabelcons(v)$ then there are two cases to consider:
  
  \begin{itemize}
    \item $v$ is $\cmark$-redundant.
      It follows that $t[p] = t[q]$ for $\{p, q\} \gets \statelabelcons(v)$.
      All sequences such that $v$ occurs in it must contain exactly the pair $(v, \cmark)$ by definition of $\cmark$-redundancy.
      We conclude that $(E(u) \cup \{\{p,q\}\}, N(u)) \models t$ holds and the term remains consistent with all extensions to $E$ and $N$ for the remaining states in the sequence.
      
    \item $v$ is $\xmark$-redundant.
      Similarly, with the observation that $(E(u), N(u) \cup \{\{p,q\}\}) \models t$.
      \qedhere
  \end{itemize}  
\end{proof}

Using Lemma~\ref{lem:wellformed-redundant} and the fact that
removing redundant states does not change the labelling of
any state we have shown that $\evalca(M, s_0, t) = \evalca(M', s_0, t)$ for all $t$.
Instead of removing redundant states after constructing the CA we could also remove them on-the-fly during the construction.
Whenever a pair of positions $\{p, q\}$ \emph{would} result in a $\cmark$-redundant (or $\xmark$-redundant) state we could instead change the parameters $N$ and $E$ to avoid such a choice.
Namely, whenever choice $\{p, q\}$ would result in a $\cmark$-redundant state then we could update $E$ to become $E \cup \{\{p, q\}\}$, and similarly for parameter $N$ in case of a $\xmark$-redundant state.
This means that the construction essentially continues as if this choice had already been made, thus avoiding the creation of redundant states.

If we consider Figure~\ref{figure:example_ca} again it follows from transitivity that the left indicated state is $\cmark$-redundant and the right indicated state $\xmark$-redundant.
If the indicated states are removed then all states of the resulting CA are reachable, which could be argued for as a form of optimum.
For transitivity it is relatively straightforward to construct a procedure to identify and remove these states. 
However, it would be more interesting to devise a method that determines all redundant states.
For example, there can also be redundancies due to the fact that a term can never be equal to any of its subterms.
Later on, we see that even more redundancies can be observed from the interleaving of matching and consistency states.

\subsection{Time and Space Complexity}
For the time complexity we consider the number of comparisons performed to determine the consistency.
Given a set of partitions $P = \{P_1, \ldots, P_n\}$ we have already established
a worst-case time complexity for naive consistency checking as being
$\functionorder{\sum_{1 \leq j \leq n} \sum_{C \in P_j} \length{C} - 1}$
with a worst-case space complexity of $\functionorder{1}$.

We establish several upper and lower bounds on the space and time complexity for CAs.
The number of comparisons performed for a given term is determined by the evaluation depth, which is the measurement for time complexity.
The upper bound is then given by $\max_{t \in \terms_\Sigma} (\evaldepth(M, t))$.
Similarly, we could define the minimal evaluation depth as the lower bound.
However, since we include final states labelled with the empty set in the CA this would simply be $1$.
Therefore, we consider the lower bound to be the minimum evaluation depth reaching any final state labelled with at least one consistency partition.
The selection function is a parameter and thus we need to consider what CA will be used for the complexity analysis.
For the upper bound on time we consider the automaton where the bound is minimised w.r.t. all possible selection function, and similarly for the upper bound on space.
On the other hand, the lower bound for both the time and space is for \emph{any} selection function.

For the upper bound of time complexity of consistency automata we can show that each pair of positions is compared at most once.
Let $m$ be the number of unique position pairs in the given partitions, where each pair of positions is counted at most once.
First, we show that the worst-case time complexity of the consistency automata evaluation is tightly bounded by $\functionorder{m}$.
This is essentially the size of $\work$ for the first call to the construction procedure, which reduces in each recursive call.

\begin{lem}
  Let $P = \{P_1, \ldots, P_n\}$ be a set of partitions and let $M$ be the optimal CA resulting from $\constructca(P, \selectca)$ for some selection function $\selectca$.
  For any term $t \in \terms_\Sigma$ the evaluation depth $\evaldepth(M, t)$ is of complexity $\functionorder{m}$, where $m$ is the number of unique position pairs in $P$.
\end{lem}

\begin{proof}
  Initially $\work = \{\{p, q\} \subseteq\in P'\}$ contains at most $m$ pairs by assumption.
  As shown in the proof of Lemma~\ref{lem:ca:construct_terminates} each choice of $\work$ can be made at most once by $\selectca$.
  Therefore, the length of the root to a final state of the resulting CA is at most $\functionorder{m}$, which is also an upper bound for number of comparisons.
\end{proof}

Furthermore, each recursive call leads to exactly two branches, which means that the size of any CA is bounded by $\functionorder{2^m}$.
The given bounds are also tight as we can construct the following example where the evaluation depth requires exactly $m$ comparisons, which in this example coincides with the number of comparisons in the naive approach.

\begin{exa}
  Let $P$ be a set of partitions $\{\{p, q\} \mid p,q \in \{1, \ldots, k\}\}$, which contains exactly $\frac{k(k-1)}{2}$ consistency classes.
  In each recursive call $\length{\work}$ decreases by exactly one, which means that it takes $\frac{k(k-1)}{2}$ comparisons to yield the set of partitions in the worst case.
  Then the automaton without removing redundant states contains $2^{\frac{k(k-1)}{2}}$ states.
\end{exa}

Therefore, it follows that the upper bound on time complexity of the consistency automata is tightly bounded up to $k^2$ position pairs, where $k$ is the number of positions.
Any additional partition, which necessarily only contains pairs of positions that have already been compared, can be inferred from the comparisons that have already been made.
Note that the upper bound on evaluation depth in this example cannot be reduced by removing redundant states, because every state on the path consisting of only $\xmark$-transitions does not contain redundant states.
Namely, it is fairly straightforward to show that there is a term $f(t_0, \ldots, t_n)$ where the subterms at two positions are equal and the subterms at any \emph{other} pair of positions are unequal.
Therefore, at any chain of $\xmark$-transitions it is still necessary to check whether the $\cmark$-transition should be taken and as such the maximum evaluation depth also remains strict for CAs where redundant states have been removed.

The lower bound on evaluation depth for CAs where redundant states have been removed is given by number of positions in the smallest consistency partition.
This bound is also strict, because we can simply give an example with a single consistency partition.
Removing redundant states by taking transivity into account is sufficient to ensure that the evaluation depth of the consistency automata never exceeds the worst-case time required for the naive consistency check.
However, the space required for using CAs is always higher than the naive consistency check.

\section{Adaptive Non-linear Pattern Matching Automata}\label{sec:anpma}

We have shown in Lemma~\ref{lem:consistency_check} that a naive matching algorithm for non-linear patterns can be obtained by using a linear matching function followed by a consistency check.  
In that case we have to check the consistency of all partitions returned by the linear matching function. 
However, as shown in the following example overlapping patterns can unify with the same prefix, but no term can match both patterns at the same time.

Consider the patterns: $\indexof{f(x, x)}{\ell_1}$ and $\indexof{f(a, b)}{\ell_2}$.
After renaming we obtain the following pairs $\ell_1 : (f(\dc_1, \dc_2), \{\{1, 2\}\})$ and $\ell_2 : (f(a, b)), \{\emptyset\})$.
Now, the resulting APMA has a final state labelled with both patterns as shown in Figure~\ref{fig:redundantapma}.
We can observe that the consistency check of positions one and two always yields false whenever the
evaluation of a term ends up in the final state labelled with $\{\ell_1, \ell_2\}$, because terms $a$ and $b$ are not equal.
Therefore, this comparison would be unnecessary.
  
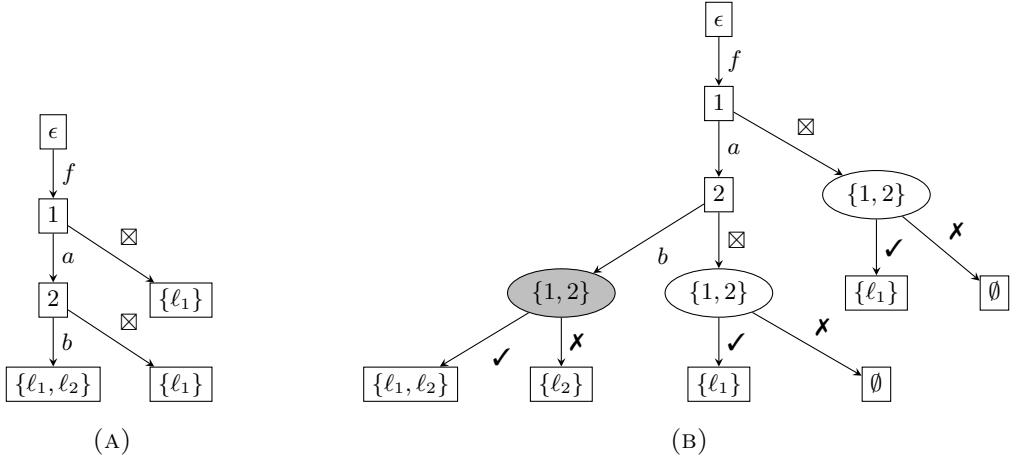
\begin{figure}
  \begin{subfigure}[t]{0.29\textwidth}
    \centering
    \input{figures/example_APMA_useless}
    \captionsetup{justification=centering}
    \caption{}\label{fig:redundantapma}
  \end{subfigure}
  \begin{subfigure}[t]{0.7\textwidth}
    \centering
    \input{figures/example_ANPMA}
    \captionsetup{justification=centering}
    \caption{}\label{fig:exampleanpma}
  \end{subfigure}
  \caption{The resulting APMA shown on the left and the corresponding ANPMA with a grey $\xmark$-redundant state on the right.}
\end{figure}
  
We could also consider an alternative where the consistency phase is performed first, but then we have the problem that whenever the given term is consistent w.r.t. partition $\{1, 2\}$ that matching on $f(a, b)$ should be avoided.
To avoid such redundancies, we propose a combination of APMAs and CAs to obtain a matching automaton for non-linear patterns called \emph{adaptive non-linear pattern matching automata}, abbreviated as ANPMAs.
The result is an automaton that has three kinds of states; match states of APMAs, consistency states of
CAs and final states, and two transition functions; one for match states and one for consistency states.

  \begin{defi}
  An adaptive non-linear pattern matching automaton (ANPMA) is a tuple
  $(\states, \delta, \statelabel, s_0)$ with
  \begin{itemize}
    \item $\states=\statesfsym\uplus\statescons\uplus\statesfin$ is a set of states where
    $\statesfsym$ is a set of match states, $\statescons$ is a set of consistency states
    and $\statesfin$ is a set of final states;

    \item $\delta=\deltafsym\uplus\deltacons$
        is a partial transition function with $\deltafsym:\statesfsym\times\bF\rightharpoonup\states$
        and $\deltacons:\statescons\times\{\cmark,\xmark\}\rightarrow\states$;

    \item $\statelabel=\statelabelfsym\uplus\statelabelcons\uplus\statelabelfin$ is a state labelling function with
        $\statelabelfsym:\statesfsym\rightarrow\positions$, $\statelabelcons:\statescons\rightarrow\positions^2$ and
        $\statelabelfin:\statesfin\rightarrow\pset{\terms}$;

    \item $s_0 \in \statesfsym$ is the initial state.
  \end{itemize}
\end{defi}

We only consider ANPMAs that have a tree structure rooted in $s_0$.
Given an ANPMA $M = (\states,\delta,\statelabel,s_0)$ and a term $t$
the procedure $\match(M, s_0, t)$ defined in Algorithm~\ref{alg:anpma_eval} defines the evaluation of the ANPMA.
It is essentially the combination of the evaluation functions for the APMAs and CAs depending on the current state.
  
\begin{algorithm}
\caption{Given a state $s$ of the ANPMA $M = (\states, \delta, \statelabel, s_0)$ and a term $t$, the following algorithm computes the pattern matches of $t$ by evaluating $M$ on $t$.}\label{alg:anpma_eval}
\begin{equation*}
  \footnotesize
  \match(M, t, s) = 
  \begin{cases}
    \statelabelfsym(s) &\text{ if } s \in \statesfin \\
                     
    \match(M, t, \deltafsym(s,f)) &\text { if } s \in \statesfsym \land \delta(s,f) \neq \bot  \\
               
    \match(M, t, \deltafsym(s,\dcneq)) &\text { if } s \in \statesfsym \land \delta(s,\dcneq) \neq \bot \land \delta(s,f) = \bot\\
             
    \emptyset &\text{ if } s \in \statesfsym \land \delta(s, \dcneq) = \delta(s,f) = \bot \\
    
    \match(M, t, \deltacons(s, \cmark)) &\text{ if } s \in \statescons \land t[p] = t[q] \\
    
    \match(M, t, \deltacons(s, \xmark)) &\text{ if } s \in \statescons \land t[p] \neq t[q] \\      
  \end{cases}
 \end{equation*}
 {\footnotesize \hfill where $f = \rootsym(t[\statelabelfsym(s)])$ and $\{p, q\} = \statelabelcons(s)$}
\end{algorithm}

The construction algorithm of the ANPMA is defined in Algorithm~\ref{alg:constructanpma}.
It combines the construction algorithm of APMAs (Algorithm~\ref{alg:constructapma}) and the construction algorithm for
CAs (Algorithm~\ref{alg:ca:construct}).
The parameters that remain the same value during the recursion are the original set
$\cL$, the result of renaming $\cL_r$ and the selection function $\select$.  Next, we have the ANPMA $M$, a state $s$
and finally the current prefix $\pref$ similar to the APMA construction and the sets of position pairs $E$ and $N$ as in
the consistency automata construction.

\begin{algorithm}
  \caption{
  Given an indexed family of patterns $\cL \subseteq \terms_\Sigma \times \cI$ and an indexed family of renamed patterns $\cL_r \subseteq  \pterms \times \pset{\pset{\positions}} \times \cI$ this algorithm computes an ANPMA for $\cL$. Initially it is called with $M=(\emptyset, \emptyset, \emptyset, s_0)$,
  the initial state $s=s_0$, the prefix $\pref=\dc_\emptypos$,
  and $E=N=\emptyset$.
  }\label{alg:constructanpma}
  \begin{algorithmic}[1]
    \footnotesize
    \Procedure{\constructanpma}{$\cL, \cL_r, \select, M, s, \pref, E, N$}
      \State $\cL_r' \gets \{\indexof{(\ell, P}{i}) \in \cL_r \mid
        \text{$\ell$ unifies with $\pref$}\, \land \neg\exists C \in P: \exists \{p, q\} \in N: p, q \in C\}$
            
      \State $\workF \gets \cF(\pref)$
      \State $\workC \gets \{\{p, q\} \in \positions^2 \mid
        \{p, q\} \subseteq\in P_i \land (\indexof{\ell}{i}, \indexof{P_i}{i}) \in \cL_r' \land
            \text{ $\pref[p]$ and $\pref[q]$ are defined} \} \setminus E$
            
      \If{($\workF=\emptyset$ and $\workC=\emptyset$) or ($\cL_r' = \emptyset$)}
        \State $M\gets M[\statesfin:=\statesfin\cup\{s\},\statelabel:=\statelabel[s \mapsto
        \{\indexof{\ell}{i} \in \cL \mid \exists \ell',P:\indexof{(\ell',P)}{i} \in \cL_r'\}]]$
      \Else
        \State $\nextstate \gets \select(\workF,\workC)$
        \If{$\nextstate=\pos$ for some position $\pos$}     
          \State $M \gets M[\statesfsym \gets (\statesfsym \cup \{s\}),\statelabelfsym \gets \statelabelfsym[s\mapsto\pos]]$
          \State $F \gets \{ f \in \bF \mid \exists(\indexof{(\ell,P)}{i}) \in \cL_r': \rootsym(\ell[\pos]) = f \}$
          \For{$f \in F$}
            \State $M \gets M[\delta:=\delta[(s,f) \mapsto s']]$
                    where $s'$ is a fresh unbranded state w.r.t. $M$
            \State $M \gets \constructanpma(\cL, \cL_r, \select, M,s',\pref[\pos/f(\dc_{\pos.1},\dots,\dc_{\pos.\ar(f)})], E, N )$
          \EndFor
          \If{$\exists(\indexof{(\ell,P)}{i})\in \cL_r': \exists\pos'\leq\pos:\rootsym(\ell[\pos'])\in\vars$}
            \State $M \gets M[\delta:=\delta[(s,\dcneq) \mapsto s']]$
                  where $s'$ is a fresh unbranded state w.r.t. $M$
            \State $M \gets \constructanpma(\cL, \cL_r, \select, M,s',\pref[\pos/\dcneq], E, N)$
          \EndIf
        \ElsIf{$\nextstate = \{p, q\}$ for some pair $\{p, q\} \in \positions^2$}
          \State $M \gets M[\statescons \gets (\statescons \cup \{s\}), \statelabelcons \gets \statelabelcons[s \mapsto \{p, q\}]]$
          \State $M \gets
            \constructanpma(\cL, \cL_r, \select, M[\deltacons:=\deltacons[(s,\cmark)\mapsto s']], s', \pref, E \cup \{\{p,q\}\}, N)$ 
            
            \hspace{\algorithmicindent} where $s'$ is an unbranded state w.r.t. $M$.
          \State $M \gets
            \constructanpma(\cL, \cL_r, \select, M[\deltacons:= \deltacons[(s, \xmark)\mapsto s']], s', \pref, E, N \cup \{\{p,q\}\})$ 
            
            \hspace{\algorithmicindent} where $s'$ is an unbranded state w.r.t. $M$.
        \EndIf
      \EndIf
      \Return $M$
      \EndProcedure
  \end{algorithmic}
\end{algorithm}

First we remove the terms that do not have to be considered anymore.
These are the elements $\indexof{(\ell,P)}{i}$ from $\cL_r$ such that $P$ is inconsistent due to the pairs in $N$
and $\pref$ does not unify with $\ell$.
Obtaining work for both types of choices is almost the same as before.  However, for $\workC$ we have added
the condition that the positions must be defined in the prefix to ensure that these positions are indeed defined when
evaluating a term.  The termination condition is that both $\workF$ and $\workC$ are empty, or that the set of
patterns $\cL_r'$ has become empty.  The latter can happen when the inconsistency of two positions removes a pattern,
which could still have other positions to be matched. 

The function $\select$ is a
function that chooses a position from $\workF$ or a pair of positions from $\workC$.
Its result determines the kind of state that $s$ becomes and as such also the outgoing transitions.
If a position is selected then $s$ will become a match state and
the construction continues as in Algorithm~\ref{alg:constructapma}.
Otherwise, similar to Algorithm~\ref{alg:ca:construct} two fresh states and two outgoing transition labelled with
$\cmark$ and $\xmark$ are created, after which the parameters $E$ and $N$ are updated.

At this point we can also see that ANPMAs are a natural extension of APMAs and CAs.
In the worst case we can first select only choices in $\workF$ and whenever $\workF$ is empty we construct only consistency states.
This would be exactly the same as performing a linear matching algorithm using APMAs followed by a consistency check using CAs, which is exactly the two-phase approach.
    
\subsection{Correctness}
The ANPMA construction algorithm yields an ANPMA that is suitable to solve the matching problem for non-empty finite sets
of (non-linear) patterns.
This can be shown by combining the efforts of Theorem~\ref{thm:correctnessapma} and
Theorem~\ref{theorem:ca:correct}.

Let $\cL$ be a finite non-empty indexed family of (non-linear) patterns and let $(\cL_r, P) = \rename(\cL)$.
Suppose that $\select:
\pset{\positions} \times \pset{\positions^2} \rightarrow \positions \uplus \positions^2$ is any function such that for all sets of positions $\workF$ and position pairs $\workC$ we have that $\select(\workF, \workC) \in \workF \uplus \workC$.

We extend the auxiliary definitions for APMAs as follows. A \emph{path} to $s_n$ is a sequence with both types of labels
$(s_0, a_0),\ldots,(s_{n-1}, a_{n-1}) \in \states \times (\bF_{\dcneq} \uplus \{\cmark, \xmark\})$ such that
$\delta(s_i, a_i) = s_{i+1}$ for all $i < n$.  A position $p$ is called \emph{visible} for state $s$ iff there is a pair
$(s_i, a_i)$ in $\pathof(s)$ such that $\statelabel(s_i).i = p$ for some $1 \leq i \leq \ar(f_i)$ or $\statelabel(s) =
\emptypos$.  A state $s$ is \emph{top-down} iff $s \in \statesfsym$ and $\statelabelfsym(s)$ is visible or $s \in
\statescons$ and both positions in $\statelabelcons(s)$ are visible. State $s$ is \emph{canonical} iff
there are no two match states in $\pathof(s)$ that are labelled with the same position.  Finally we say that an ANPMA is
\emph{well-formed} iff $\statelabel(s_0)=\emptypos$, and all states are top-down and canonical.

\begin{lem}\label{lem:anpma:properconstruction}
  The procedure $\constructanpma(\cL_r, P, \select, (\emptyset,\emptyset,\emptyset,s_0), s_0, \dc_{\epsilon}, \emptyset, \emptyset)$ terminates and yields a well-formed ANPMA.
\end{lem}
\begin{proof}
  We only show that the recursion terminates. The rest is similar to the proof for Lemma~\ref{lem:properconstruction}, with the additional observation that positions in $P$ are only chosen when they are defined in the prefix.
  Given the parameters $\pref_1, E_1,N_1$
  and $\pref_2,E_2,N_2$ we can fix the ordering:
  \begin{align*}
    &(\pref_1\strbelow\pref_2 \wedge E_1=E_2 \wedge N_1=N_2) \ \vee \\
    &(\pref_1=\pref_2 \wedge E_1\subset E_2 \wedge N_1=N_2) \ \vee \\
    &(\pref_1=\pref_2 \wedge E_1=E_2 \wedge N_1\subset N_2)\,.
  \end{align*}
  The prefixes are again only defined on positions that are defined in patterns of $\cL$
  and the sets $E$ and $N$ are bounded by a finite product of positions, hence the ordering is well-founded.
  The recursive calls conform to to this ordering; therefore the recursion terminates.
\end{proof}

\noindent Let $M = (\states, \delta, \statelabel, s_0)$ be the ANPMA resulting from $\constructanpma(\cL, \select)$. Let $t \in \terms_\Sigma$ be a term and $\cL_t$ be equal to $\{\indexof{\ell}{i} \in \cL \mid \ell \leq t\}$. For every state $s \in \states$ we define $\cL(s)$ to be equal to $\{\indexof{\ell}{i} \in \cL \mid \indexof{(\ell', P)}{i} \in \cL'_r(s_i)\}$.  We show that the evaluation algorithm on $M$ satisfies a number of invariants. 

\begin{lem}\label{lem:anpma:subset}
  For all $s \in \states$ such that $\match(M, t, s_0) = \match(M,t,s)$ it holds that: (a) $(E(s_i), N(s_i)) \models t$, (b) $\cL_t \subseteq \cL(s)$ and (c) if $s \in \statesfin$ then $\statelabel(s_f)=\cL_t$.
\end{lem}
\begin{proof}
  Take an arbitrary term $t$.
  We prove the first two invariants by induction on the length of $\pathof(s)$.
  
  Base case, the empty path and as such $s = s_0$. 
  $E(s_0) = N(s_0) = \emptyset$ and $\cL_t \subseteq \cL$, and $\cL = \cL(s_0) = \cL(s)$, as such the statements hold vacuously.
  
  Inductive step.
  Let $s$ be an arbitrary state and suppose that the statements hold for $\match(A, t, s_0) = \match(A, t, s)$.
  Suppose $\match(A, t, s) = \match(A, t, s')$ for some $s' = \delta(s, x)$ such that $x \in (\bF_{\dcneq} \uplus \{\cmark, \xmark\}))$.
  Now, there are two cases to consider:
  \begin{itemize}
    \item $s \in \statescons$.    
      Let $\{p, q\}$ be the value of $\statelabelcons(s_k)$.
      Again, there are two cases to consider:
      \begin{itemize}
        \item $t[p] = t[q]$ in which case $E(s')$ is $E(s) \cup \{p,q\}$ and $N(s') = N(s)$.
          Therefore, $(E(s'), N(s')) \models t$ holds.  
          Furthermore, $\cL(s') = \cL(s)$ because also $\pref(s') = \pref(s)$.
        
        \item Otherwise, $t[p] \neq t[q]$ in which case $N(s')$ is equal to $N(s) \cup \{p, q\}$ and $E(s') = E(s)$.
        Therefore, $(E(s'), N(s')) \models t$ holds.
        Consider any $\indexof{\ell}{i} \in \cL(s)$ such that $\indexof{\ell}{i}  \notin \cL(s')$.
        From $\pref(s') = \pref(s)$ it follows that for $\indexof{(\ell', P)}{i} \in \cL_r$ it holds that $P$ is not consistent w.r.t. $t$ by observation that positions $\{p, q\} \subseteq\in P$ are included in $N$ and $t[p] \neq t[q]$.
        Therefore, by Lemma~\ref{lem:rename} it holds that $\indexof{\ell}{i} \notin \cL_t$.        
      \end{itemize}
  
    \item $s \in \statesfsym$.
       It holds that $E(s') = E(s)$ and $N(s') = N(s)$.
       Therefore, $(E(s'), N(s')) \models t$ remains true.
       Now, we can use the same argument as before to argue that any pattern removed must not unify with $\pref(s')$.
       Then the same arguments as given in Lemma~\ref{lem:matchapmainvariant} can be used to show that $\cL_t \subseteq \cL(s')$ holds.
  \end{itemize}
   
  Finally, if $s \in \statesfin$ from the fact that $\statelabel(s) = \cL(s)$ we know that $\cL_t \subseteq \statelabel(s)$.
  It only remains show that $\statelabel(s_f) \subseteq \cL_t$.
  There are two cases for this state to become a final state during construction:
  \begin{itemize}
    \item Both $\workC = \emptyset$ and $\workF = \emptyset$.
      Suppose for a contradiction that there is some $\indexof{\ell}{i} \in \statelabel(s_f)$ such that $\indexof{\ell}{i} \notin \cL_t$.
      It follows that $\ell \nleq t$, which means that for $\indexof{(\ell', P)}{i} \in \cL_r$ that $\ell' \not\below t$ or $t$ is not consistent w.r.t. $P$ by Lemma~\ref{lem:rename}.
      We show that both cases lead to a contradiction:
      \begin{itemize}
        \item Case $\ell' \not\below t$.
          This follows essentially from the same observations as Lemma~\ref{lem:validmatches}.
          
        \item Case $t$ is not consistent w.r.t. $P$.
          From the fact that $\pref(s_f)$ unifies with $t$ and that it is a ground term due to $\workF = \emptyset$ it follows that for all $p, q$ such that $\{p, q\} \subseteq\in P_i$ they are defined in $\pref(s)$ and therefore it holds that $\{p, q\} \in E(s)$.
          Therefore, for all $p, q \in C$ for consistency class $C \in P_i$ it holds that $t[p] = t[q]$ and as such $t$ is consistent w.r.t. $P_i$.  As such $\index{\ell}{i}$ is not an element of $\statelabel(s_f)$, contradicting our
          assumption.
      \end{itemize} 
    
    \item The set $\cL(s)$ is empty.
     In this case $\statelabel(s_f)$ is empty and $\statelabel(s_f) \subseteq \cL_t$ by definition.
     \qedhere
  \end{itemize}
\end{proof}

\begin{lem}\label{lem:anpma:uniquestate}
  If $\cL_t = \emptyset$ then $\match(M, t, s_0) = \emptyset$.
\end{lem}
\begin{proof}
  We show that $\match(M, t, s_0) \neq \statelabel(s)$ for all final states $s$ for which $\statelabel(s) \neq \emptyset$.
  Let $s_f$ be an arbitrary final state such that $\statelabel(s_f) \neq \emptyset$ and pick some pattern $\indexof{\ell}{i} \in \statelabel(s_f)$.
  By assumption $\ell\not\below t$ and by Lemma~\ref{lem:rename} it holds for the pair $\indexof{(\ell', P)}{i} \in \cL_r$ that $\ell' \not\below t$ or $t$ is not consistent w.r.t. $P$.
  \begin{itemize}
    \item If $\ell' \not\below t$  then by Proposition~\ref{prop:altmatching} it follows that there is a position $p$ and a function symbol $f\in\bF$ such that
      $\hd(\ell[p])=f$ and $\hd(t[p])\neq f$.
      By Lemma~\ref{lem:finalstates} it must be that $\hd(\pref(s)[p])=f$,
      by which there must be a pair $(s_i,f)\in\pathof(s)$.
      Since $\match$ is a function we again have that $\match(M,t,s_0)=\match(M,t,s_i)=\match(M,t,s_f)$.
      However, by its definition we know that $\hd(t[p])=f$, which contradicts the assumption that $\indexof{\ell}{i} \in \statelabel(s_f)$.  
      
   \item If $t$ is not consistent w.r.t. $P$.   
    By Lemma~\ref{lem:anpma:subset} we know that $(E(s_f), N(s_f)) \models t$ and for all pairs $\{p, q\} \subseteq\in P$ it holds that $\{p, q\} \in E$ for $\workC$ to become empty, because all positions of pattern $\ell$ are defined in the prefix $\pref(s_f)$.
    As such $t$ must be consistent w.r.t. $P$, which contradicts the assumption that $\indexof{\ell}{i} \in \statelabel(s_f)$.  
    \qedhere
  \end{itemize}
 
\end{proof}

\begin{thm}
  Then $\lambda t.\match(M, t, s_0)$ is a matching function for $\cL$.
\end{thm}

\begin{proof}
  If $\cL_t$ is empty then by Lemma~\ref{lem:anpma:uniquestate} we get that $\match(M, t, s_0)=\emptyset=\cL_t$ as required.
  Otherwise, we have that $\match(M, t, s_0) = \match(M, t, s_f)$ for some final state $s_f$.
  Then by the definition of $\match$ and Lemma~\ref{lem:anpma:subset} we conclude
  $\match(M, t, s_0)=\match(M, t, s_f) = L(s_f) = \cL_t$.
\end{proof}

\subsection{The redundancy problem for ANPMAs}
We have seen the notion of $f$-redundant states and dead states for APMAs in Section~\ref{sec:apma}.
In particular Lemma~\ref{lem:apmanonredundancy} shows that the construction by Sekar et al. eliminates all such
redundancies.
We have also discussed the notion of $\cmark$-redundancy and $\xmark$-redundancy
in Section~\ref{sec:ca}.
We repeat the notions here.

\begin{defi}
We observe the following redundancies for ANPMAs.
\begin{itemize}
\item
Given $f\in\bF_\dcneq$, we say that a consistency state $s$ is \emph{$f$-redundant} iff
for all terms $t$, whenever $\match(M,t,s_0)=\match(M,t,s)$, then
$\match(M,t,s)=\match(M,t,\delta(s,f))$.

\item
A consistency state $s$ is $\cmark$-redundant iff, for all terms $t$,
whenever $\match(M,t,s_0)=\match(M,t,s)$, then $\match(M,t,s)=\match(M,t,\delta(s,\cmark))$.

\item
A consistency state $s$ is $\xmark$-redundant iff, for all terms $t$,
whenever $\match(M,t,s_0)=\match(M,t,s)$, then $\match(M,t,s)=\match(M,t,\delta(s,\xmark))$.

\item We say that a state $s$ is \emph{dead} iff
for all terms $t$, whenever $\match(M,t,s_0)=\match(M,t,s)$, then
$\match(M,t,s)=\emptyset$.
\end{itemize}
\end{defi}

There is a connection between the pairs in $E$ and $N$,
and the observed function symbols that are represented by $\pref$.
For example, if one has (partial) knowledge of the function symbols in $t_1$
and one also knows that $t_1=t_2$, then the same knowledge of the function symbols of $t_2$
follows by definition of term equality.
The other way around we have that full knowledge of the function symbols of both $t_1$ and $t_2$
makes the equality check between them redundant.
And even though checking the functions symbols of $t_1$ first and then comparing $t_1$ with $t_2$ to
determine a match,
we also need to be able to do this the other way around for completeness of the optimisation.

The removal of redundant states in ANPMAs is a difficult problem in its full generality.
We describe some observations and some examples in the remainder of this section,
but we leave the full problem open for future work.

Suppose that we have a recursive call in the construction with the parameters
$\cL_r, \pref, E$ and $N$.
Consider a position $p\in\cF(\pref)$.
Suppose that we can derive that a match state $s$ with $L(s)=p$,
would be $f$-redundant during the construction.
Then $p$ can be removed from $\workF$.
From Lemma~\ref{lem:apmanonredundancy} it follows that removing redundant positions
from $\workF$ is easy in the absence of consistency partitions.
But when we have partial information about term (un)equality recorded in $E$ and $N$,
then it could be that the function symbols of many positions in $\workF$ are already known.

\begin{exa}\label{ex:fsxsy}
Consider the patterns $\ell_1:f(x,x)$ and $\ell_2:f(s^{999}(x),s^{999}(y))$ where $s^{999}(x)$
denotes the application of 999 unary successor symbols to the variable $x$.
In the lucky case, we can detect that $f(t_1,t_2)$ satisfies $t_1=t_2$.
The best strategy to also detect a match for $\ell_2$ would be to check whether
$t_1$ matches $s^{999}(x)$.
This means that $t_2$ matches $s^{999}(y)$ as well,
so almost half the work of checking the function symbols can be skipped.
\end{exa}

Suppose we can derive that
a consistency state $s$ with $L(s)=\{p,q\}$ will be $\cmark$-redundant.
Then $(p,q)$ can be removed from $\workC$ as well.

\begin{exa}\label{ex:fgaga}
Consider the patterns $\ell_1:f(x,x)$, $\ell_2:f(g(a),y)$ and $\ell_3:f(x,g(a))$
and suppose that the selection function prioritises checking all function symbols.
Then a consistency check for the term $f(g(a),g(a))$ would be redundant,
since we already have full knowledge of all function symbols.
\end{exa}

Suppose we can derive that
a consistency state $s$ with $L(s)=\{p,q\}$ will be $\xmark$-redundant.
Then the partition that gave rise to the comparison of positions $p$ and $q$
can be removed since it is inconsistent.
Since this partition could have given rise to pairs in $\workC$,
and also to positions in $\workF$,
these sets should be computed again with a reduced set of partitions and patterns.

\begin{exa}
Consider again the patterns of Example~\ref{ex:fgaga}.
The term $f(g(a),b)$ only matches $\ell_2$.
A consistency check to rule out $\ell_1$ is redundant after checking all function symbols,
because a mismatching function symbols has already been detected.
\end{exa}


These ideas could be captured in a procedure that replaces the first three lines in the ANPMA construction.
We expect that it is possible to define a procedure that computes minimal sets $\workF$ and $\workC$,
along with a smaller pattern set $\cL_r'$ from which all derivable inconsistent patterns have been removed.
Note that to ensure correctness this procedure also has to adapt the parameters $\pref$, $E$ and $N$ internally in a similar way as the algorithm.

\begin{conj}\label{conj:optimal}
There is an algorithm \textsc{RemoveRedundancies} that takes the parameters $\cL_r, \pref, E$ and $N$,
and yields a (reduced) pattern set $\cL_r'\subseteq\cL_r$ and two sets $\workF$ and $\workC$ such that:
replacing lines 2-4 of $\constructanpma$ by
\[(\cL_r',\workF,\workC)\gets \textsc{RemoveRedundancies$(\cL_r,\pref,E,N)$}
\,,\]
makes
$\constructanpma$ yield a correct ANPMA without redundant states,
for every pattern set $\cL_r$ and every selection function $\select$.
\end{conj}

Note that the comparison in Example~\ref{ex:fsxsy} could also be useful in the case of \emph{only} the pattern $\ell_2:f(s^{999}(x),s^{999}(y))$. However, the current technique does not allow us to choose arbitrary positions to compare.
Therefore, it could also be interesting to either extend the given patterns with new patterns to allow these choices, or to extend the technique to choose arbitrary positions for consistency checks.
Similarly, it might be useful to add a comparison state where one of the pairs is not a position, but rather a fixed term without variables.
This could be used to exploit constant time comparisons to further improve the efficiency of the matching procedure, which would be especially useful for switch statements over natural numbers. 

We now focus on a final example to show that identifying these redundant states can be non-trivial.
However, this example shows that interleaving the matching and consistency states can be advantageous in practice.
Consider the following patterns: $\ell_1: f(x, x)$, $\ell_2: f(x, f(x, y))$, $\ell_3: f(x, f(y, x))$, $\ell_4: f(f(x, y), x)$ and $\ell_5: f(f(y, x), x)$.
These patterns can occur as part of an optimisation where applications of an expensive operation $f$ can be avoided.
For example, if the $f$ operator implements set union then $f(x, x)$ represents the case where
the set union is computed of two equivalent sets.
Similarly $f(x, f(x, y))$ represents the case where set union if applied to a set that already contains the set $x$.

These patterns are non-linear and therefore we first apply the $\rename$ function to obtain the corresponding linear pattern and consistency partitions.
This result in the following renamed pattern set.

\begin{tabular}{l l}
  $\ell_1 : (f(\dc_1, \dc_2), \{\{1,2\}\})$ & \\ 
  $\ell_2 : (f(\dc_1, f(\dc_{2.1}, \dc_{2.2})), \{\{1, 2.1\}\})$ & $\ell_3 : (f(\dc_1, f(\dc_{2.1}, \dc_{2.2})), \{\{1, 2.2\}\})$ \\ 
  $\ell_4 : (f(f(\dc_{1.1}, \dc_{1.2}), \dc_{2}), \{\{1.1, 2\}\})$ & $\ell_5 : (f(f(\dc_{1.1}, \dc_{1.2}), \dc_{2}), \{\{1.2, 2\}\})$
\end{tabular}

If we consider the APMA for the set of linear patterns: $f(\dc_1, \dc_2)$, $f(\dc_1, f(\dc_{2.1}, \dc_{2.2}))$ and $f(f(\dc_{1.1}, \dc_{1.2}), \dc_2)$ then we can easily see that it takes three steps to reach any of the final states.
Namely, we check the positions $\epsilon$, one and two for the occurrence of the head symbol $f$ and then we obtain the subset of patterns that match, which can be either $\{\ell_1\}, \{\ell_1, \ell_2, \ell_3\}$, $\{\ell_1, \ell_4, \ell_5\}$ or $\{\ell_1, \ell_2, \ell_3, \ell_4, \ell_5\}$.
After that a consistency automaton or naive consistency check requires exactly one comparison for each element in the set to determine the set of consistent partitions.

However, the consistency automaton contains some redundancies.
For example, whenever the partition of $\ell_1$ is consistent w.r.t. a term $t$ it cannot be the case that any of the other partitions are consistent w.r.t. term $t$, because a term cannot be equal to any of its subterms.
Formally, this means that $t \neq t[i]$ for any natural number $i$, which can be proven by structural induction on terms.
We can generalise this statement such that for any positions $p$ and $r \neq \epsilon$ it holds that $t[p.r] \neq t[p]$.

This observation can be used to avoid the previously mentioned redundancy.
If for any term $t$ it holds that $t[1] = t[2]$ then only partition $\{1,2\}$ can be consistent with respect to $t$.
Furthermore, if we know that for a term $t[1.1] = t[2]$ then $t[1.1.1] = t[2.1]$ by definition of equality and since $t[1] \neq t[1.1.1]$ by the observation above it follows that $t[1] \neq t[2.1]$.
Therefore, we can conclude that if $t[1.1] = t[2]$ then only patterns $\ell_4$ and $\ell_5$ can match.
Similarly, if  $t[1] = t[2.1]$ then only patterns $\ell_2$ and $\ell_3$ can match.
Using these observations we can construct the ANPMA without redundancies that is shown in Figure~\ref{fig:practical_anpma}.

\begin{figure}[ht]
\input{figures/example_practical_anpma.tex}
\caption{The pruned ANPMA for the patterns $\ell_1: f(x, x)$, $\ell_2: f(x, f(x, y))$, $\ell_3: f(x, f(y, x))$, $\ell_4: f(f(x, y), x)$ and $\ell_5: f(f(y, x), x)$.}\label{fig:practical_anpma}
\end{figure}

This ANPMA is both smaller in the number of states when compared to an APMA followed by individual CAs.
Furthermore, its evaluation depth for all terms that match the patterns $\ell_1, \ell_2$ and $\ell_3$ is strictly smaller than without interleaving.
For the patterns $\ell_4$ and $\ell_5$ the evaluation depth remains the same.
Therefore, this example shows that interleaving the two phases and removing redundancies can yield a more efficient non-linear matching procedure in practice.

\section{Conclusion and Future Work}\label{sec:conclusion}
In this paper, we presented a formal proof for the correctness of APMAs.
Furthermore, we introduced  as a 
deterministic automaton to perform the consistency checking, from which some
redundant states could be removed by taking the previous choices into account.
These two automata are then combined to obtain an ANPMA which could be evaluated by only performing comparisons and taking the
corresponding outgoing edge.

ANPMAs offer a formal platform to study the relations between linear pattern matching and consistency checking.
There are still some questions that have arisen from this work.
As mentioned in the previous section, the current ANPMA construction algorithm can contain redundant states.
We expect that there is an optimisation function as described in Conjecture~\ref{conj:optimal}
that takes care of this problem.
But it might just as well be the case that this problem is undecidable altogether.

Secondly we did not study selection functions in this work.
All three automaton construction algorithms in this paper are parametrised in a selection function
that decides for each node what will happen next.
We have shown that all constructions yield correct automata for every selection function,
with the side note that the selection indeed yields an element from its input set.
The size of all three kinds of automata depends heavily on the selection function that is used.
For APMAs some selection functions have already been studied in \cite{SekarRR95:adaptive}.

Finally, it would be interesting to implement this approach.
This work is a theoretical approach to ultimately reduce the number of steps that matching requires in for example a term rewrite engine.
However, many of the existing pattern matching problems do not support non-linear pattern matching as discussed in the introduction.
It would be interesting to find out whether exploiting $\cO(1)$ term equality checking is worth the construction time and in particular the identification of redundant states in practice.

\section*{Acknowledgements}

We would like to thank Jan Friso Groote, Bas Luttik and Tim Willemse for their feedback and discussion.
This work was supported by the TOP Grants research programme with project number 612.001.751 (AVVA),
which is (partly) financed by the Dutch Research Council (NWO).

\bibliographystyle{alpha}
\bibliography{bibliography}

\appendix

\section{Proof details of Section~\ref{sec:apma}}
\subsection{Proof details of Lemma~\ref{lem:matchapmainvariant}}
\begin{proof}
  By induction on the length of $\pathof(s)$.
  If there are no pairs in $\pathof(s)$ then it must be that $s=s_0$.
  From $\pref(s_0)=\dc_\epsilon$ it follows that $\cL(s_0)=\cL$.
  Then the base case follows from $\cL_t\subseteq\cL=\cL(s_0)=\cL(s)$.
  
  Let $s$ be an arbitrary state and suppose that $\match(M, t, s_0) = \match(M, t, s)$ and
  assume the induction hypothesis $\cL_t\subseteq\cL(s)$.
  Now suppose $\match(M, t, s)=\match(M, t, s')$
  where $s'=\delta(s,f)$ for some $f\in\bF_{\dcneq}$
  and let $\statelabel(s)=p$.
  \begin{itemize}
  \item If $f\in\bF$ then $\pref(s')=\pref(s)[p/f(\dc_{p.1},\dots,\dc_{p.\ar(f)})]$.
  By definition of $\match$ we know that $\hd(t[p])=f$.
  
  Let $\ell\in\cL_t$. We show that $\ell$ unifies with $\pref(s')$.
  We know that $\ell\below t$ by assumption.
  From the induction hypothesis it follows that $\ell$ unifies with $\pref(s)$.
  So there is a term $u$ such that $\ell\below u$ and $\pref(s) \below u$.
  Then we distinguish two cases.
  \begin{itemize}
  \item If $\ell[p']$ is a variable for some $p'\posleq p$ then
  $\ell$ unifies with $\pref(s')$.
  \item If $\hd(\ell[p])$ is a function symbol then by $\ell\below t$ it must be that
  $\hd(\ell[p])=f$, so $\ell$ unifies with $\pref(s')$.
  \end{itemize}
  
  \item If $f=\dcneq$ then $\pref(s')=\pref(s)[p/\dcneq]$.
  By definition of $\match$ we know that $\delta(s,\hd(t[p]))$ is undefined.
  
  From the construction algorithm we then know that there is no pattern $\ell\in\cL(s)$ such that
  $\hd(\ell[p])\in\bF$ and there is at least one pattern $\ell\in\cL(s)$ such that $\ell[p']$
  is a variable for some position $p'\posleq p$.
  
  Let $\ell\in\cL_t$.
  By induction hypothesis we know that $\ell$ unifies with $\pref(s)$.
  We show that $\ell$ unifies with $\pref(s')$ by showing that $\ell[p']=\dc_{p'}$
  for some position $p'\posleq p$.
  \begin{itemize}
  \item Suppose that $\ell[p]$ exists.
  Since $\ell\below t$ and $\hd(t[p])\neq\hd(\ell[p])$ it must be that $\ell[p]=\dc_p$.
  
  \item Suppose that $\ell[p]$ does not exist.
  Pick the lowest position $p'$ such that $p'\posle p$ and $\ell[p']$ exists
  and assume for a contradiction that $\hd(\ell[p'])=f$ for some function symbol $f$.
  Then it must be that
  $\hd(\pref(s)[p'])=f$ by the induction hypothesis.
  However, $\pref(s)[p]$ exists and from $p'\posle p$
  it follows that $\ell[p]$ has subterms of the function symbol $f$,
  which contradicts the assumption that $p'$ is the lowest position strictly higher than $p$.
  So $\ell[p']=\dc_{p'}$.
  \qedhere
  \end{itemize}
  \end{itemize}
\end{proof}

\subsection{Proof details of Lemma~\ref{lem:uniquestate}}
\begin{proof}
  \begin{enumerate}[a)]
  \item We show that $\match(M, t, s_0)\neq\statelabel(s)$ for all final states $s$.
  Let $s_f$ be an arbitrary final state and pick some pattern $\ell\in\statelabel(s_f)$.
  By assumption $\ell\not\below t$ and by Proposition~\ref{prop:altmatching} it follows that
  there is a position $p$ and a function symbol $f\in\bF$ such that
  $\hd(\ell[p])=f$ and $\hd(t[p])\neq f$.
  By Lemma~\ref{lem:finalstates} it must be that $\hd(\pref(s)[p])=f$,
  by which there must be a pair $(s_i,f)\in\pathof(s)$.
  Since $\match$ is a function we have $\match(M,t,s_0)=\match(M,t,s_i)=\match(M,t,s_f)$.
  Then, by definition of $\match$ we obtain $\hd(t[p])=f$, a contradiction.
    
  \item
    Let $\ell\in\cL_t$.
    We prove that for all $s$ such that $\match(M,s_0,t)=\match(M,s,t)$,
    we have that $\delta(s,\hd(t[\statelabel(s)]))$ or $\delta(s,\dcneq)$ is defined.
  
    Suppose that $\match(M,s_0,t)=\match(M,s,t)$.
    From Lemma~\ref{lem:matchapmainvariant} it follows that $\ell\in\cL(s)$.
    If $\hd(\ell[\statelabel(s)])=f$ for some function symbol $f$
    then the construction algorithm created an $f$-transition to a new state,
    by which $\delta(s,f)$ exists.
    Otherwise if $\hd(\ell[\statelabel(s)])$ does not exist
    then by $\ell\below t$ there must be a position $p\posle\statelabel(s)$ such that $\ell[p]=\dc_p$.
    In that case a $\dcneq$-transition is created and hence $\delta(s,\dcneq)$ exists.
  
    By definition of $\match$ we then have that $\match(M,s_0,t)$ cannot yield the empty set,
    so it must terminate in a final state.\qedhere
  \end{enumerate}
\end{proof}

\subsection{Proof details of Lemma~\ref{lem:validmatches}}
\begin{proof}
  Since $\statelabel(s_f)=\cL(s_f)$ we know
  that $\cL_t\subseteq\statelabel(s_f)$ by Lemma~\ref{lem:matchapmainvariant}.
  It only remains show that $\statelabel(s_f)\subseteq\cL_t$.
  Since $s_f$ is a final state we have that $\statelabel(s_f)=\{\ell\in\cL\mid\ell\below\pref(s_f)\}$.
  Suppose for a contradiction that there is some $\ell\below\pref(s_f)$ such that
  $\ell\not\below t$.
  Then there is a position $p$ such that $\hd(\ell[p])\in\bF$ and
  $\hd(t[p])\neq\hd(\ell[p])$.
  We have $\hd(\ell[p])=\hd(\pref(s_f)[p])$ by assumption.
  So, there is a pair $(s_i,f_i)$ in $\pathof(s_f)$
  such that $\statelabel(s_i)=p$.
  By definition of $\match$ we then have $\hd(t[p])=f_i=\hd(\ell[p])$,
  a contradiction.\qedhere
\end{proof}

\end{document}

%% file: figures/example_introANPMA.tex
\begin{tikzpicture}[description/.style={fill=white,inner sep=2pt}]
\scriptsize
\matrix (m) [matrix of math nodes, row sep=1.5em,
column sep=1.5em, text height=1.5ex, text depth=0.25ex]
{  
                              &                      & |[draw,rectangle]| \epsilon  & \\
                              &                      & |[draw,ellipse]| \{1,2\}     & \\  
                              & |[draw,rectangle]| 1 &                              & |[draw,rectangle]| 2 \\  
|[draw,rectangle]| \{\ell_1\} & |[draw,rectangle]| \{\ell_1,\ell_3\} &              & |[draw,rectangle]| 1 \\ 
                              &                      &                              & |[draw,rectangle]| \{\ell_2\}           \\   
};

\path[-stealth]
  (m-1-3) edge node[right] {$f$} (m-2-3)  
  (m-2-3) edge node[above] {$\cmark$} (m-3-2)  
  (m-2-3) edge node[above] {$\xmark$} (m-3-4)  
  
  (m-3-2) edge node[above] {$\dcneq$} (m-4-1)  
  (m-3-2) edge node[right] {$a$} (m-4-2)  
  (m-3-4) edge node[right] {$b$} (m-4-4)  
  (m-4-4) edge node[right] {$a$} (m-5-4)  
;
\end{tikzpicture}

%% file: figures/example_APMA.tex
\begin{tikzpicture}[description/.style={fill=white,inner sep=2pt}]
\footnotesize
\matrix (m) [matrix of math nodes, row sep=2em,
column sep=0.75em, text height=1.5ex, text depth=0.25ex]
{
 & |[draw,rectangle]| \epsilon & & \\
 & |[draw,rectangle]| 2 & & \\
 & |[draw,rectangle]| 1 & & \\
|[draw,rectangle]| 3 & & |[draw,rectangle]| 3 & \\
|[draw,rectangle]| \{f(a, b, x)\} & |[draw,rectangle]|  \{f(c, b, x), f(c, b, c)\} & &
|[draw,rectangle]|  \{f(c, b, x)\} \\
};

\path[-stealth]
  (m-1-2) edge node[right] {$f$} (m-2-2)
  (m-2-2) edge node[right] {$b$} (m-3-2)
  (m-3-2) edge node[above left] {$a$} (m-4-1)
  (m-3-2) edge node[above right] {$c$} (m-4-3)
  (m-4-1) edge node[right] {$\dcneq$} (m-5-1)
  (m-4-3) edge node[above] {$c$} (m-5-2)
  (m-4-3) edge node[above right] {$\dcneq$} (m-5-4)
;
\end{tikzpicture}

%% file: figures/example_CA_bad.tex
\begin{tikzpicture}[description/.style={fill=white,inner sep=2pt}]
\footnotesize
\matrix (m) [matrix of math nodes, row sep=1em,
column sep=0.75em, text height=1.5ex, text depth=0.25ex]
{
  &                            & |[draw,ellipse]| \{1, 2\} &                        & \\          
  & |[draw,ellipse]| \{1, 3\}    &                         & |[draw,ellipse]| \{1, 3\}  & \\
  & |[draw,ellipse,fill=lightgray]| \{2, 3\} & |[draw,ellipse,fill=lightgray]| \{2, 3\} & |[draw,ellipse]| \{P_2\} & |[draw,ellipse]| \emptyset & \\
  |[draw,ellipse]| \{P_1, P_2, P_3\} & |[draw,ellipse]| \{P_1, P_2\} & |[draw,ellipse]| \{P_1\} & |[draw,ellipse]| \{P_1\} & \\
};

\path[-stealth]
 (m-1-3) edge node [above left]  {$\cmark$} (m-2-2)
 (m-1-3) edge node [above right] {$\xmark$} (m-2-4)
 (m-2-2) edge node [left]        {$\cmark$} (m-3-2)
 (m-2-2) edge node [above right] {$\xmark$} (m-3-3)
 (m-2-4) edge node [left]        {$\cmark$} (m-3-4)
 (m-2-4) edge node [above right] {$\xmark$} (m-3-5) 
 (m-3-2) edge node [above left]  {$\cmark$} (m-4-1)
 (m-3-2) edge node [right]       {$\xmark$} (m-4-2)
 (m-3-3) edge node [left]        {$\cmark$} (m-4-3)
 (m-3-3) edge node [above right] {$\xmark$} (m-4-4)
;
\end{tikzpicture}

%% file: figures/example_CA_select_bad.tex
\usetikzlibrary{shapes}

\begin{tikzpicture}[description/.style={fill=white,inner sep=2pt}]
\scriptsize
\matrix (m) [matrix of math nodes, row sep=1em,
column sep=0.75em, text height=1.5ex, text depth=0.25ex]
{
 & |[draw,ellipse]| \{2, 3\} & & \\          
 & |[draw,ellipse]| \{1, 2\} & & & |[draw,ellipse]| \{1, 2\}\\
 & |[draw,ellipse]| \{3, 4\} & |[draw,ellipse]| \{3, 4\} & & |[draw,ellipse]| \{3, 4\} & |[draw,ellipse]| \emptyset \\
 |[draw,ellipse]| \{P_1, P_2\} & |[draw,ellipse]| \{P_2\} & |[draw,ellipse]| \{P_1\} & |[draw,ellipse]| \emptyset & |[draw,ellipse]| \{P_1\} & |[draw,ellipse]| \emptyset\\
};

\path[-stealth]
 (m-1-2) edge node [left]        {$\cmark$} (m-2-2)
 (m-1-2) edge node [above right] {$\xmark$} (m-2-5)
 (m-2-2) edge node [left] {$\cmark$} (m-3-2)
 (m-2-2) edge node [above right] {$\xmark$} (m-3-3)
 (m-3-2) edge node [above left] {$\cmark$} (m-4-1)
 (m-3-2) edge node [left] {$\xmark$} (m-4-2)
 (m-3-3) edge node [left] {$\cmark$} (m-4-3)
 (m-3-3) edge node [above right] {$\xmark$} (m-4-4)
 (m-2-5) edge node [left] {$\cmark$} (m-3-5)
 (m-2-5) edge node [above right] {$\xmark$} (m-3-6)
 (m-3-5) edge node [left] {$\cmark$} (m-4-5)
 (m-3-5) edge node [above right] {$\xmark$} (m-4-6)
;
\end{tikzpicture}

%% file: figures/example_CA_select_good.tex
\usetikzlibrary{shapes}

\begin{tikzpicture}[description/.style={fill=white,inner sep=2pt}]
\scriptsize
\matrix (m) [matrix of math nodes, row sep=1em,
column sep=0.75em, text height=1.5ex, text depth=0.25ex]
{
 & |[draw,ellipse]| \{1, 2\} & & \\          
 & |[draw,ellipse]| \{2, 3\} & |[draw,ellipse]| \emptyset  & \\
 & |[draw,ellipse]| \{3, 4\} & |[draw,ellipse]| \{3, 4\} & \\
 |[draw,ellipse]| \{P_1, P_2\} & |[draw,ellipse]| \{P_2\} & |[draw,ellipse]| \{P_1\} & |[draw,ellipse]| \emptyset \\
};

\path[-stealth]
 (m-1-2) edge node [left]        {$\cmark$} (m-2-2)
 (m-1-2) edge node [above right] {$\xmark$} (m-2-3)
 (m-2-2) edge node [left] {$\cmark$} (m-3-2)
 (m-2-2) edge node [above right] {$\xmark$} (m-3-3)
 (m-3-2) edge node [above left] {$\cmark$} (m-4-1)
 (m-3-2) edge node [left] {$\xmark$} (m-4-2)
 (m-3-3) edge node [left] {$\cmark$} (m-4-3)
 (m-3-3) edge node [above right] {$\xmark$} (m-4-4)
;
\end{tikzpicture}

%% file: figures/example_APMA_useless.tex
\begin{tikzpicture}[description/.style={fill=white,inner sep=2pt}]
\footnotesize
\matrix (m) [matrix of math nodes, row sep=2em,
column sep=2em, text height=1.5ex, text depth=0.25ex]
{
 |[draw,rectangle]| \epsilon & \\
 |[draw,rectangle]| 1 & \\
 |[draw,rectangle]| 2 & |[draw,rectangle]| \{\ell_1\} \\
 |[draw,rectangle]| \{\ell_1, \ell_2\} & |[draw,rectangle]| \{\ell_1\} \\
};

\path[-stealth]
  (m-1-1) edge node[right] {$f$} (m-2-1)
  (m-2-1) edge node[right] {$a$} (m-3-1)
  (m-2-1) edge node[above right] {$\dcneq$} (m-3-2)
  (m-3-1) edge node[right] {$b$} (m-4-1)
  (m-3-1) edge node[above right] {$\dcneq$} (m-4-2)
;
\end{tikzpicture}

%% file: figures/example_ANPMA.tex
\begin{tikzpicture}[description/.style={fill=white,inner sep=2pt}]
\footnotesize
\matrix (m) [matrix of math nodes, row sep=2em,
column sep=2em, text height=1.5ex, text depth=0.25ex]
{
 & & |[draw,rectangle]| \epsilon & \\
 & & |[draw,rectangle]| 1 & \\
 & & |[draw,rectangle]| 2 & |[draw,ellipse]| \{1,2\} \\
 & |[draw,ellipse,fill=lightgray]| \{1,2\} & |[draw,ellipse]| \{1,2\} & |[draw,rectangle]| \{\ell_1\} & |[draw,rectangle]| \emptyset  \\
 |[draw,rectangle]| \{\ell_1, \ell_2\} & |[draw,rectangle]| \{\ell_2\} & |[draw,rectangle]| \{\ell_1\} & |[draw,rectangle]|  \emptyset \\
};

\path[-stealth]
  (m-1-3) edge node[right] {$f$} (m-2-3)
  (m-2-3) edge node[right] {$a$} (m-3-3)
  (m-2-3) edge node[above right] {$\dcneq$} (m-3-4)
  (m-3-3) edge node[below right] {$b$} (m-4-2)
  (m-3-3) edge node[right] {$\dcneq$} (m-4-3)
  (m-3-4) edge node[right] {$\cmark$} (m-4-4)
  (m-3-4) edge node[above right] {$\xmark$} (m-4-5)
  (m-4-3) edge node[right] {$\cmark$} (m-5-3)
  (m-4-3) edge node[above right] {$\xmark$} (m-5-4)
  (m-4-2) edge node[below right] {$\cmark$} (m-5-1)
  (m-4-2) edge node[right] {$\xmark$} (m-5-2)
;
\end{tikzpicture}

%% file: figures/example_practical_anpma.tex
\begin{tikzpicture}[description/.style={fill=white,inner sep=2pt}]
\scriptsize
\matrix (m) [matrix of math nodes, row sep=1.5em,
column sep=1.5em, text height=1.5ex, text depth=0.25ex]
{
 & & |[draw,rectangle]| \epsilon & \\
 & & |[draw,ellipse]| \{1,2\} & \\
 & |[draw,rectangle]| \{\ell_1\} & |[draw,rectangle]| 1 & \\
 & & |[draw,ellipse]| \{1.1,2\} & |[draw,rectangle]| 2 \\
 & |[draw,ellipse]| \{1.2,2\} & |[draw,ellipse]| \{1.2,2\} & & |[draw,ellipse]| \{1,2.1\} \\
 |[draw,rectangle]| \{\ell_4, \ell_5\} & |[draw,rectangle]| \{\ell_4\} & |[draw,rectangle]| \{\ell_5\} & |[draw,rectangle]| 2 &  |[draw,ellipse]| \{1, 2.2\} & |[draw,ellipse]| \{1, 2.2\} \\
 & & |[draw,ellipse]| \{1, 2.1\} & |[draw,rectangle]| \{\ell_2, \ell_3\} & |[draw,rectangle]| \{\ell_2\} & |[draw,rectangle]| \{\ell_3\} & |[draw,rectangle]| \emptyset \\ 
  & |[draw,ellipse]| \{1, 2.2\} & |[draw,ellipse]| \{1, 2.2\} \\
 |[draw,rectangle]| \{\ell_2, \ell_3\} & |[draw,rectangle]| \{\ell_2\} & |[draw,rectangle]| \{\ell_3\} & |[draw,rectangle]| \emptyset \\
};

\path[-stealth]
  (m-1-3) edge node[right] {$f$} (m-2-3)
  (m-2-3) edge node[above left] {$\cmark$} (m-3-2)
  (m-2-3) edge node[right] {$\xmark$} (m-3-3)
  
  (m-3-3) edge node[right] {$f$} (m-4-3)
  (m-3-3) edge node[above right] {$\dcneq$} (m-4-4)
  
  (m-4-3) edge node[above left] {$\cmark$} (m-5-2)
  (m-4-3) edge node[right] {$\xmark$} (m-5-3)
  (m-4-4) edge node[above right] {$f$} (m-5-5)
  
  (m-5-2) edge node[above left] {$\cmark$} (m-6-1)
  (m-5-2) edge node[right] {$\xmark$} (m-6-2)
  (m-5-3) edge node[right] {$\cmark$} (m-6-3)
  (m-5-3) edge node[above right] {$\xmark$} (m-6-4)
  (m-5-5) edge node[right] {$\cmark$} (m-6-5)
  (m-5-5) edge node[above right] {$\xmark$} (m-6-6)
  
  (m-6-4) edge node[above left] {$f$} (m-7-3)
  (m-6-5) edge node[above left] {$\cmark$} (m-7-4)
  (m-6-5) edge node[right] {$\xmark$} (m-7-5)
  (m-6-6) edge node[right] {$\cmark$} (m-7-6)
  (m-6-6) edge node[above right] {$\xmark$} (m-7-7)
  
  (m-7-3) edge node[above left] {$\cmark$} (m-8-2)
  (m-7-3) edge node[right] {$\xmark$} (m-8-3)
  
  (m-8-2) edge node[above left] {$\cmark$} (m-9-1)
  (m-8-2) edge node[right] {$\xmark$} (m-9-2)
  (m-8-3) edge node[right] {$\cmark$} (m-9-3)
  (m-8-3) edge node[above right] {$\xmark$} (m-9-4)
  
;
\end{tikzpicture}